\documentclass[11pt, letterpaper]{article}
\usepackage[margin=1in]{geometry}

\usepackage{amsmath,amsthm,amssymb,amsfonts}
\usepackage[utf8]{inputenc}
\usepackage[american]{babel}
\usepackage{graphicx}
\usepackage[url=false,giveninits=true,maxbibnames=99,style=alphabetic,maxalphanames=5]{biblatex}
\usepackage{mathtools}
\usepackage{enumitem}
\usepackage{csquotes}
\usepackage[noend]{algpseudocode}
\usepackage{algorithm}
\usepackage{xparse}
\usepackage{xspace}
\usepackage{color}
\usepackage{caption}
\usepackage{subcaption}
\usepackage{fullpage}
\usepackage{placeins}
\usepackage{xcolor}
\usepackage{makecell}
\usepackage{qcircuit}
\usepackage{thmtools}
\usepackage{thm-restate}
\usepackage[textsize=scriptsize]{todonotes}
\usepackage{verbatim}
\usepackage{qcircuit}
\usepackage{comment}
\usepackage{mathdots}

\usepackage{hyperref}
\hypersetup{
    pdftitle={The Information Complexity of Decision Trees}, 
    pdfauthor={Anonymous submission}, 
    colorlinks=true, 
    linkcolor=blue, 
    citecolor=cyan, 
    urlcolor=cyan 
}
\usepackage[capitalise, nameinlink]{cleveref}
\crefname{equation}{}{}
\usepackage{crossreftools}

\usepackage{dsfont}

\setuptodonotes{color=blue!15}

\newcommand\cfont\mathsf

\protected\def\verythinspace{%
  \ifmmode
    \mskip0.5\thinmuskip
  \else
    \ifhmode
      \kern0.08334em
    \fi
  \fi
}

\newcommand{\be}{\begin{equation}}
\newcommand{\ee}{\end{equation}}

\renewcommand{\epsilon}{\varepsilon}

\DeclarePairedDelimiterX\braket[2]{\langle}{\rangle}{#1 \delimsize\vert #2}
\DeclarePairedDelimiterX\ketbra[2]{\lvert}{\rvert}{#1 \delimsize\rangle\delimsize\langle #2}

\setlist[itemize]{noitemsep, topsep=0pt}
\setlist[enumerate]{noitemsep, topsep=0pt}

\declaretheorem[numberwithin=section]{theorem}

\declaretheorem[sibling=theorem]{corollary}
\declaretheorem[sibling=theorem]{lemma}
\declaretheorem[sibling=theorem]{proposition}

\declaretheorem[sibling=theorem]{fact}

\declaretheorem[sibling=theorem,style=definition]{definition}

\crefname{observation}{observation}{observations}
\Crefname{observation}{Observation}{Observations}

\makeatletter
\newcommand{\subalign}[1]{%
  \vcenter{%
    \Let@ \restore@math@cr \default@tag
    \baselineskip\fontdimen10 \scriptfont\tw@
    \advance\baselineskip\fontdimen12 \scriptfont\tw@
    \lineskip\thr@@\fontdimen8 \scriptfont\thr@@
    \lineskiplimit\lineskip
    \ialign{\hfil$\m@th\scriptstyle##$&$\m@th\scriptstyle{}##$\hfil\crcr
      #1\crcr
    }%
  }%
}
\NewDocumentCommand{\LeftComment}{s m}{%
  \Statex \IfBooleanF{#1}{\hspace*{\ALG@thistlm}}\(\triangleright\) #2}

\def\moverlay{\mathpalette\mov@rlay}
\def\mov@rlay#1#2{\leavevmode\vtop{%
   \baselineskip\z@skip \lineskiplimit-\maxdimen
   \ialign{\hfil$\m@th#1##$\hfil\cr#2\crcr}}}
\newcommand{\charfusion}[3][\mathord]{
    #1{\ifx#1\mathop\vphantom{#2}\fi
        \mathpalette\mov@rlay{#2\cr#3}
      }
    \ifx#1\mathop\expandafter\displaylimits\fi}

\makeatother

\algnewcommand{\LineComment}[1]{\State \(\triangleright\) #1}

\algblockdefx[ON]{Blk}{EndBlk}[1]
  {#1}
  {}

\makeatletter
\ifthenelse{\equal{\ALG@noend}{t}}%
  {\algtext*{EndBlk}}
  {}%
\makeatother

\AtEveryBibitem{%
  \clearlist{language}%
}

\def\({\left(}
\def\){\right)}

\newcommand{\E}{\mathop{\mathbb{E}}}

\newcommand{\rX}{\mathsf{X}}

\newcommand{\rI}{\mathsf{I}}
\newcommand{\CC}{\mathsf{CC}}

\newcommand{\IC}{\mathsf{IC}}
\newcommand{\info}{\mathsf{info}}

\def\anon{0} 
\def\conf{0} 

\title{The Information Complexity of Decision Trees}
\date{}

\if\anon1
\author{Anonymous submission}
\else
\author{Avantika Agarwal\\
    \small Institute for Quantum Computing\\
    \small University of Waterloo\\
    \small \texttt{avantika.agarwal@uwaterloo.ca}
    \and
    Shalev Ben{-}David\\
    \small Institute for Quantum Computing\\
    \small University of Waterloo\\
    \small \texttt{shalev.b@uwaterloo.ca}
    \and
    Eric Blais\\
    \small University of Waterloo\\
    \small \texttt{eric.blais@uwaterloo.ca}
}
\fi

\begin{document}

\maketitle

\begin{abstract}
    We define and study a measure of information complexity for randomized decision trees. We prove three main results about this complexity measure:

    \paragraph{Information equals amortized size complexity.} We show that the information complexity of randomized decision tree is equal to the logarithm of the amortized worst-case randomized tree size complexity of computing a function $f$. That is, when computing $f$ on $n$ inputs, the logarithm of the randomized tree size is exactly equal to the amount of information needed to compute the function. 
    
    \paragraph{Information allows for tree size compression.} We show that even when computing $f$ on a single input, the information complexity can be used to compress the size of a tree, if we allow a small loss in success probability. With the recent characterization of Chattopadhyay, Dahiya, Mande, Radhakrishnan, and Sanyal (2023), this result shows that the depth of AND-OR trees can also be compressed in terms of information complexity.
    
    \paragraph{Direct Product Theorems.} We show that the success-conditioned variant of information complexity satisfies a perfect direct product theorem. This result gives an information complexity analogue of the direct product theorem for success-conditioned randomized query complexity by Ben-David and Blais (2025).
\end{abstract}

\if\conf0
    \clearpage
    {\small\tableofcontents}
    \clearpage
\fi

\section{Introduction}

A fundamental observation with surprisingly far-reaching consequences in the study of communication complexity is that it can be useful to measure the cost of a communication protocol in terms of \emph{information} (that the players learn about each other's input or that an external observer learns about all the inputs through observing the transcript) instead of just counting the number of bits that have been transmitted. This observation leads to the notion of information complexity that has resulted in a number of notable advances in communication complexity. For great introductions to this area, see the surveys~\cite{Bra23,Wei15} and the book~\cite{RY20}.

The idea of measuring information instead of counting bits applies equally well to the setting of query complexity.
Indeed, consider the \emph{randomized query complexity} of the function $f \colon \{0,1\}^k \to \{0,1\}$ with respect to the distribution $\mu$ over $\{0,1\}^k$ defined to be the minimum expected number of queries that a randomized algorithm must make to its input $x$ to output the value $f(x)$ with success probability at least $\frac23$.
By representing a randomized query algorithm as a distribution $R$ over deterministic decision trees, we can define this complexity measure as
\[
\overline{\sf R}^\mu(f) = \min_{R \,:\, {\sf success}^\mu_f(R) \ge \frac23} \E_{T \sim R; X \sim \mu}[ {\sf depth}(T, X)],
\]
where ${\sf depth}(T,X)$ denotes the depth of the leaf of the decision tree $T$ that we reach when we traverse it down following the values of the bits of $X$.

With a nearly identical definition except that we measure the amount of information that the algorithm learns about its input instead of just counting the number of indices of the input that it queries, we obtain a natural notion of \emph{information complexity}.
This complexity measure can again be defined by using the decision tree representation of query algorithms as
\[
\IC^\mu(f) = \min_{R \,:\, {\sf success}^\mu_f(R) \ge \frac23} {\rm I}( X; L \mid T),
\]
where ${\rm I}( X; L \mid T)$ is the conditional mutual information of the input $X \sim \mu$ and the leaf $L$ of the tree $T \sim R$ reached by traversing it according to $X$.
This amount also equals the entropy of the sequence of bits revealed by the queries of the algorithm during its execution.

Given the success of the use of information complexity in the study of communication complexity (not to mention the numerous other remarkably effective applications of information theory in many other domains!) it is most natural to ask: 

\newtheorem*{mainQ}{Main Question}
\begin{mainQ}
What can the information complexity measure $\IC^\mu$ tell us about the query complexity of Boolean functions?
\end{mainQ}

While information theory has been used extensively in the study of query complexity (see, notably, \cite{JKS03} and \cite{GLSS23}) and although
variants of this question have been raised repeatedly within the query complexity community,\footnote{See for example the Stack Exchange question~\cite{Yue15}. Anecdotally, we have also had multiple discussions around this question with different sets of colleagues over the years.} very little appears to be known about this question beyond the obvious fact that information complexity provides a lower bound on randomized query complexity. Yet, the study of information complexity within communication complexity offers some evidence that we should be able to obtain more insightful answers. Notably, in communication complexity, the following is known:

\begin{description}
\item[Information is Amortized Complexity.] The amortized (per-input) communication complexity of a function when it is computed on multiple inputs in parallel is exactly captured by (internal) information complexity:
$\lim_{n \to \infty} \frac{\CC^{\mu^{n}}(f^{n})}{n} = \IC_i^\mu(f)$~\cite{BR11}.

\item[Compression.] 
The internal information complexity of a function is always bounded above by its communication complexity and the maximum gap that can exist between the two measures is exponential~\cite{Bra15,GKR16}. But we can also say more: a communication protocol with information cost $I$ and communication cost $C \gg I$ can always be ``compressed''  to yield an alternative bounded-error protocol with communication cost 
$O( \sqrt{I \cdot C} )$~\cite{BBCR13}.

\item[Direct Product Theorems.] 
Information complexity can establish \emph{Direct Product Theorems} that give lower bounds on the communication complexity of computing $n$ instances of a function even when the required probability of success on all $n$ instances is exponentially small: $\CC_{e^{-\Omega(n)}}^{\mu^n}(f^n) = \tilde{\Omega}(n \, \IC_i^\mu(f))$~\cite{BRWY13}.
\end{description}

Na\"ively, one may hope to obtain analogues of these results in query complexity by simply replacing the measure ${\sf CC}^\mu(f)$ with the randomized query complexity measure $\overline{\sf R}^\mu(f)$ and using the above query-based notion of information complexity in place of the internal information complexity used in communication complexity.
But such hopes are quickly shown to be unreasonable when we consider the ${\sf OR}_k \colon \{0,1\}^k \to \{0,1\}$ function that outputs the OR of its $k$ input variables and the distribution $\mu$ that puts $\frac12$ probability mass on the input $0^k$ and probability mass $\frac1{2k}$ on each of the $k$ inputs that contain a single $1$. 

The randomized query complexity of the OR function on the above distribution $\mu$ is $\overline{\sf R}^\mu({\sf OR}_k) = \Theta(k)$, but its information complexity is $\IC^\mu({\sf OR}_k) = O(\log k)$. This already shows an exponential gap between the two measures.
And a closer analysis of the OR function shows even more: the natural protocol that queries the bits of the input in order until it finds a 1 achieves the optimal randomized query complexity and information complexity of the OR function simultaneously, so we cannot hope to obtain a compression result for $\overline{\sf R}^\mu$. And, in fact, it appears that the information complexity does not characterize the amortized query complexity of the OR function or give good direct product theorems either.
Taken together, these observations show that any meaningful link between the information complexity measure ${\sf IC}^\mu$ and randomized query complexity (if any exists!) must go through some other complexity measure, not $\overline{\sf R}^\mu$.

\subsection{Our Results}

We show that the information complexity measure $\IC^\mu$ \emph{does} lead to strong amortization, compression, and direct product theorems in randomized query complexity, but in terms of the \emph{tree size} complexity of Boolean functions.

Recall that we can associate every randomized query algorithm with a distribution over decision trees. The \emph{size} of a decision tree $T$, denoted ${\sf size}(T)$, is the number of leaves that it contains. The properties of decision tree size complexity measures have been studied previously in \cite{CDMRS23, Sni85, JRSW99, CMP25, CDM23, Che25}.

The \emph{distributional worst-case randomized size complexity} of $f \colon \{0,1\}^k \to \{0,1\}$ with respect to the distribution $\mu$ over $\{0,1\}^k$ and the success parameter $\gamma \in (\frac12, 1]$ is
\[
\mathsf{RSize}^\mu_{\gamma}(f) = \min_{R \,:\,\mathsf{success}^\mu_f(R) \geq \gamma} \max_{T \in {\sf supp(R)}} \mathsf{size}(T),
\]
the maxiumum size $s$ such that every randomized decision tree that computes $f$ with success probability at least $\gamma$ on inputs drawn from $\mu$ must have at least $s$ leaves in one of the trees in its support.\footnote{We remark that the distributional worst-case randomized tree size complexity is equal to its deterministic counterpart ${\sf DSize}^\mu_{\gamma}(f)$; this equivalence, however, does not hold with the multi-input extension of the worst-case size complexity measure that we will consider shortly.}

The logarithm of the worst-case size complexity is related to the randomized query complexity of Boolean functions via the inequality
$\log {\sf RSize}^\mu_\gamma(f) \le {\sf R}^\mu_\gamma(f)$, where ${\sf R}^\mu_\gamma(f)$ is the minimum worst-case depth complexity of randomized decision trees that compute $f$ with success probability at least $\gamma$ over $\mu$. 
But, as the OR function shows, the two sides of this inequality can be separated by an exponential factor.
The logarithm of the worst-case size complexity of $f$ also gives an upper bound on its information complexity (informally, because the maximum amount of entropy of the set of bits revealed by a decision tree is bounded above by the number of leaves in the tree).
And, as we show below, there are also much deeper connections between information complexity and tree size complexity.

\subsubsection*{Information is Amortized Size Complexity}

We first show that information complexity equals amortized worst-case size complexity.
Let $f^n \colon \{0,1\}^{k \times n} \to \{0,1\}^n$ denote the function obtained by applying $f$ on $n$ different inputs.
A randomized decision tree computes $f^n$ with \emph{per-input success probability} at least $\gamma$ with respect to the product distribution $\mu^n$ when the $i$th bit of its output equals the value of $f$ on the $i$th input with probability at least $\gamma$ for each $i \in [n]$ on inputs drawn from $\mu^n$.
Let $\mathsf{RSize}^{\mu^n}_{\gamma^{\otimes n}}(f^n)$ denote the maximum value of $s$ such that every randomized decision tree that computes $f^n$ with per-input success probability at least $\gamma$ with respect to $\mu^n$ must have a decision tree of size $s$ in its support.
Then we obtain the following identity.

\begin{restatable}[Amortization for Tree Size]{theorem}{amortization}
\label{thm:amortization}
    For any function $f: \{0,1\}^k \rightarrow \{0,1\}$, distribution $\mu$ on $\{0,1\}^k$ and success probability $\gamma < 1$,
    \begin{align*}
        \lim_{n \rightarrow \infty} \frac{\log\left(\mathsf{RSize}^{\mu^n}_{\gamma^{\otimes n}}(f^n)\right)}{n} = \IC^\mu_\gamma(f)
    \end{align*}
\end{restatable}

The proof of \cref{thm:amortization} is obtained by proving inequalities in both directions.
The lower bound direction relies on a key structural property of information complexity: it behaves additively, in the following way.

\begin{restatable}[Additivity of Information Complexity]{lemma}{additivity}
\label{lem:additivity}
    For any function $f: \{0,1\}^k \rightarrow \{0,1\}$, distribution $\mu$ on $\{0,1\}^k$ and success probability $\gamma$,
    \begin{align*}
        \IC_{\gamma^{\otimes n}}^{\mu^n}(f^n) = n\,\IC_{\gamma}^{\mu}(f)
    \end{align*}
\end{restatable}

The additivity property of information complexity is known to hold in communication complexity~\cite{BR11}. \cref{lem:additivity} shows that the same property holds in the query complexity setting as well.\footnote{\cite{Sur24} showed that the amortized worst-case randomized query complexity is equal to the expected randomized query complexity in the non-distributional setting. This result can be viewed as a different analogue of the ``information is amortized communication complexity'' result but it is incomparable with \cref{thm:amortization}.}

\subsubsection*{Information Allows for Tree Size Compression}

\Cref{thm:amortization} says that it is possible to compress the tree size in the amortized setting where a single tree computes the value of $f$ on many inputs at once.
But it does not say anything about using the information complexity of a function to possibly compress the size of a tree that computes $f$ on only a single input.
Our next result shows that compression is possible in this setting as well, and that in fact we get a very strong connection between size complexity and information complexity when we allow for a small loss in success probability.

\begin{restatable}[Compression for Tree Size]{theorem}{compression}
\label{thm:compression}
    For any function $f: \{0,1\}^k \rightarrow \{0,1\}$, distribution $\mu$ on $\{0,1\}^k$, success probability $\gamma$ and $0 < \epsilon \leq \gamma - \frac{1}{2}$,
    \begin{align*}
        \log\left(\mathsf{RSize}^{\mu}_{\gamma - \epsilon}(f)\right) \leq \left(\frac{2}{\epsilon}\right)^2\IC_{\gamma}^{\mu}(f)
    \end{align*}
\end{restatable}

Since $\IC_\gamma^\mu(f) \le \log {\sf RSize}^\mu_\gamma(f)$ for every function $f$, \cref{thm:compression} shows that the information complexity of $f$ in the bounded-error regime is essentially equivalent to the log of its size complexity. We can also use \cref{thm:compression} to obtain a direct sum theorem for randomized tree size, see \cref{cor:sizeds} for a precise statement.

\subsubsection*{Direct Product Theorems}

Next, we consider the Direct Product Theorem question for information complexity.
Here, we again consider the function $f^n$ on $n$ inputs, but we relax the success condition of the algorithms that compute this function to succeed on all inputs with probability $\gamma^n$ instead of requiring the per-input success probability to be $\gamma$ on each of the $n$ inputs.
To see why this is a weaker requirement on algorithms, consider the ``lazy algorithm'' that does nothing and simply guesses the output with probability $1-\gamma^n$ and computes $f^n$ exactly with the remaining probability $\gamma^n$. This algorithm has overall success probability (slightly more than) $\gamma^n$, but its per-input success probability is also only $\approx \gamma^n \ll \gamma$.
The relaxed correctness condition on the algorithms means that proving lower bounds for these algorithms is more challenging and, indeed, doing so has been the focus of much research in query complexity~\cite{Sha01, KSW04, She12, Dru12, BB25}. 

Does a Direct Product Theorem hold for information complexity? The lazy algorithm described above shows that in general $\IC_{\gamma^n}^{\mu^n}(f^n) \le \gamma^n \cdot n \, \IC_1^\mu(f)$, which can be exponentially smaller than the bound $n \, \IC_\gamma^\mu(f)$ that we would like to obtain in a direct product theorem.
This limitation is due to the fact that information complexity is an average-case complexity measure, and a similar limitation also holds for average-case randomized query complexity~\cite{BB25}.
And, as in~\cite{BB25}, this limitation can be overcome by considering \emph{success-conditioned} information complexity.
We define this measure ${\sf sIC}$ formally in \cref{sec:icinquery}, but we can think of it informally as the expected amount of information about the input learned by the algorithm conditioned on the algorithm succeeding in computing the correct value of the function. 

Our next result shows that success-conditioned information complexity does satisfy a Direct Product Theorem.

\begin{restatable}[Direct Product Theorem for Success-Conditioned Information Complexity]{theorem}{sicdpt}
\label{thm:dpt}
    For every function $f: \{0,1\}^k \rightarrow \{0,1\}$, distribution $\mu$ on $\{0,1\}^k$, success probability $\gamma$ and $n \geq 1$,
    \begin{align*}
        \mathsf{sIC}^{\mu^n}_{\gamma^n}(f^n) = \Omega\left(n\cdot\mathsf{sIC}^{\mu}_{\gamma}(f)\right).
    \end{align*}
\end{restatable}

\Cref{thm:dpt} immediately implies a number of other direct product theorems. 
For example, it implies that that the distributional randomized query (or depth) complexity of every function $f$ satisfies
\[
        \mathsf{R}^{\mu^n}_{\gamma^n}(f) \geq \Omega(n \cdot \IC^{\mu}_{\gamma}(f)).
\]
It also recovers stronger direct product theorems, such as Drucker's direct product theorem for randomized tree size complexity~\cite{Dru12}.

\begin{corollary}[Direct Product Theorem for Randomized Tree Size]
\label{cor:dpt-treesize}
    For any function $f: \{0,1\}^k \rightarrow \{0,1\}$ and success probability $\gamma = \frac{2}{3}$,
    \begin{align*}
         \Omega\left(n \, \log\left(\mathsf{RSize}_{\gamma}(f)\right)\right)\leq \log\left(\mathsf{RSize}_{\gamma^{n}}(f^n)\right) \leq n\log\left(\mathsf{RSize}_{\gamma}(f)\right)
    \end{align*}
\end{corollary}

See \cref{cor:fulltreedirectproduct} for the precise statement of the most general direct product theorem with respect to size complexity that we obtain from \cref{thm:dpt}.

\subsubsection*{Information and Query Complexity, Revisited}

The results above show that the information complexity of a Boolean function is closely connected to the log of its randomized size complexity, not its randomized query complexity measure.
But could the randomized query complexity of a function $f$ be characterized by the information complexity of a closely-related function?
We show that, indeed, the use of a simple parity gadget suffices to relate information complexity back to randomized query complexity directly.

\begin{restatable}{theorem}{ictoquery}
    For any function $f: \{0,1\}^k \rightarrow \{0,1\}$,
    \begin{align*}
        \max_{\mu}\IC^\mu(f \circ \oplus_2) = \Theta(\mathsf{R}(f)).
    \end{align*}
\end{restatable}

Finally, consider the AND-OR decision tree model where each node of the tree can query the AND or the OR of any subset of the bits of the input. Chattopadhyay, Dahiya, Mande, Radhakrishnan, and Sanyal~\cite{CDMRS23} established a tight connection between the randomized query complexity of functions in this model and the randomized size complexity of functions.
Using this connection, we obtain an alternative characterization of the information complexity of functions. For example, we obtain the following compression result.

\begin{corollary}[Compression for AND-OR query complexity]
    For any function $f: \{0,1\}^k \rightarrow \{0,1\}$, distribution $\mu$ on $\{0,1\}^k$, success probability $\gamma$ and $0 < \epsilon \leq \gamma - \frac{1}{2}$,
    \begin{align*}
\mathsf{R^{\wedge \vee}}^{\,\mu}_{\gamma - \epsilon}(f) \leq 4 \cdot \left(\frac{2}{\epsilon}\right)^2\IC_{\gamma}^{\mu}(f)
    \end{align*}
\end{corollary}
\begin{proof}
    \cite{CDMRS23} showed that $\mathsf{R^{\wedge \vee}}^{\,\mu}_{\gamma - \epsilon}(f) \leq 4 \cdot \log\left(\mathsf{RSize}^{\mu}_{\gamma - \epsilon}(f)\right)$. Combined with \cref{thm:compression}, we get the claim.
\end{proof}
As in the case of \cref{thm:compression}, this result is essentially tight and the randomized AND-OR complexity of $f$ is essentially asymptotically equivalent in the bounded-error regime to its information complexity, up to a multiplicative $\log k$ term.

\subsection{Proof Overviews}

A key observation underlying all of our proofs is that there is a way to define the information complexity of functions on a per-leaf basis. As we show in \cref{sec:icinquery}, by defining the cost of a leaf $\ell$ in a deterministic decision tree to be proportional to log of the probability that a random input drawn from $\mu$ leads to that leaf $\ell$, we obtain a natural measure of ``surprisal'' where $\IC^\mu_\gamma$ is the minimum expected surprisal of the leaf that we reach in a randomized decision tree that computes $f$ with success probability at least $\gamma$ on inputs drawn from $\mu$. 

The surprisal viewpoint of information complexity leads to particularly natural proofs of the amortization and compression theorems (\cref{thm:amortization,thm:compression}). 
An overview of these proofs can be found at the beginning of \cref{sec:amortization,sec:compression}, respectively.

The proof of the direct product theorem (\cref{thm:dpt}) is more involved. We provide a high-level overview of this proof below.

\paragraph{Direct Product Theorems.} 
The starting point of the proof of the direct product theorem (\cref{thm:dpt}) is the proof of the optimal direct product theorem for randomized query complexity of \cite{BB25}. The key ingredient in that proof is a notion of \emph{discounted score}. The two main technical components of the proof are a perfect direct product theorem (also known as a tensorization lemma) for this notion of discounted score and an equivalence lemma between discounted score and (success-conditioned) query complexity.

The notion of discounted score introduced in \cite{BB25} is the expected value of a per-input measure that assigns each input in the domain of the function a value that combines the score of a tree on this input (which we can think as 1 if the tree outputs the correct function value on this input and 0 otherwise) with the \emph{cost} of the tree on this input, measured in terms of the number of queries that the tree makes on this particular input. As presented in \cite{BB25}, the notion of discounted score and the proofs of the tensorization and equivalence lemmas present two limitations that prevent it from being applied directly to information complexity.

The first limitation is the per-input nature of the cost measure used in the definition of discounted score and in the proofs of the lemmas. With information complexity, there is no natural per-input measure of information cost. But this limitation is minor; we show that discounted score can also be defined as a per-leaf measure instead; see \cref{sec:discountedscore} for the details.

The second limitation in the proof of \cite{BB25} is more significant: the proof relies on the cost measure being the query complexity measure, a measure that is very well structured in that it increases by 1 for each step in a path from the root to a leaf of the tree. But for information complexity, the corresponding per-leaf cost measure does not satisfy this property.
We get around this second limitation in two steps. In the first step, we generalize the original proof of \cite{BB25} and show that the tensorization and equivalence lemmas hold for all discounted score notions built on generalized cost measures that satisfy three properties: normalization, monotonicity, and additivity across instances. In the second step, we show that the surprisal cost measure associated with information complexity satisfies these three properties and enables us to complete the proof of the direct product theorem.

\section{Preliminaries}
\label{sec:prelims}
We start this section by reviewing some basic notational conventions and standard inequalities we use in the proofs. In \cref{sec:entropy} we describe some basic preliminaries from information theory that are used for our definitions and proofs. In \cref{sec:depthandsize}, we define formally the various notions of randomized query and tree size complexity that we consider in our work. In \cref{sec:icinquery}, we formally define the information complexity and success-conditioned information complexity of a function, and show some basic properties.

Given a randomized decision tree $R$ computing some function $f: \{0,1\}^k \rightarrow \{0,1\}$ on an input distribution $\mu$, we define $L(\mu, R)$ to be the distribution obtained on the leaves of $R$, when a deterministic tree $D \sim R$ is sampled from the internal randomness of $R$, and when the input is sampled from $\mu$. Given an input $x$ and a deterministic tree $D$, we use $D(x)$ to denote the output of the leaf of $D$ reached on input $x$. We denote the success probability of a randomized algorithm $R$ computing a function $f$ on the input distribution $\mu$ as
\begin{align*}
    \mathsf{success}^{\mu}_f(R) = \E_{D \sim R, x \sim \mu}[D(x) = f(x)]
\end{align*}
When we are interested in the worst-case success probability of an algorithm $R$ computing a function $f$, we denote it as
\begin{align*}
    \mathsf{success}_f(R) = \min_{x \in \mathsf{dom}(f)} \E_{D \sim R}[D(x) = f(x)]
\end{align*}
We now review some basic facts we use in the paper.
\begin{fact}[Lemma 3.2 from \cite{Dru12}]\label{fact:indepinputs}
    Consider the random variable $X = X_1 \times \ldots \times X_n$ corresponding to the input distribution $\mu^n$. Conditioned on seeing query responses $q = q_1 \times \ldots \times q_n$, we have that $X|_q = X_1|_{q_1} \times \ldots \times X_n|_{q_n}$, i.e, the $n$ inputs are independent, and the $i^{th}$ input is distributed according to the random variable $X_i|_{q_i}$.
\end{fact}

\begin{fact}[Chernoff Bound]\label{fact:chernoff}
    Let $\mathbf{X_1}, \cdots, \mathbf{X_n}$ be $n$ independent random variables between $0$ and $1$. If $\mathbf{X} = \sum_{i=1}^n \mathbf{X_i}$ and $\mu = \E[\mathbf{X}]$, then for any $\delta > 0$,
    \begin{align*}
        \Pr[\mathbf{X} \geq (1+\delta)\mu] &\leq e^{-\frac{\delta^2}{2+\delta}\cdot\mu} \\
        \Pr[\mathbf{X} \leq (1-\delta)\mu] &\leq e^{-\frac{\delta^2}{2}\cdot\mu}
    \end{align*}
\end{fact}

\begin{fact}[Kraft's Inequality]\label{fact:kraft}
    There exists a binary tree such that it has $l_k$ leaves at depth $k$ for $k \in \mathbb{N}$ if and only if $\sum_{k \in \mathbb{N}} \frac{l_k}{2^k} \leq 1$.
\end{fact}

\begin{fact}[Jensen's Inequality]\label{fact:jensen}
    For any finite set $X = \{x_1, \ldots, x_n\}$ where $x_i \in \mathbb{R}$, a distribution $\mu$ on $X$ and a convex function $f: X \rightarrow \mathbb{R}$, we have that
    \begin{align*}
        f\left(\E_{\mu}[X]\right) \leq \E_{\mu}[f(X)]
    \end{align*}
\end{fact}

\begin{fact}\label{fact:convexityofe}
    For a random variable $y \in [0, d]$ and a distribution $\mu$ on $y$,
    \begin{align*}
        \E_{\mu}[e^{-y}] \leq 1- \frac{1-e^{-d}}{d} \E_{\mu}[y] \leq e^{-\frac{1-e^{-d}}{d} \E_{\mu}[y]}
    \end{align*}
\end{fact}
\begin{proof}
    Using convexity of $e^{-y}$, we get that for every $y \in [0, d]$, setting $y = (1-\frac{y}{d}) \cdot 0 + \frac{y}{d} \cdot d$,
    \begin{align*}
        e^{-y} \leq (1-\frac{y}{d}) \cdot e^{-0} + \frac{y}{d} \cdot e^{-d} = 1- \frac{1-e^{-d}}{d} \cdot y
    \end{align*}
    Therefore, the first inequality follows from linearity of expectation. For the second inequality, we will show that for any constant $c > 0$, $e^{-cx} \geq 1-cx$ for $x \geq 0$. To show this, consider the function $f(x) = e^{-cx} - 1 + cx$. Then $f(0) = 0$ and $f'(x) = c(1-e^{-cx}) \geq 0$ for all $x \geq 0$. Therefore, $f(x) \geq 0$ for all $x \geq 0$.
\end{proof}

\subsection{Shannon Entropy and Mutual Information}\label{sec:entropy}
In this section, we review the basics of information theory. The Shannon entropy of a random variable $X$ measures the amount of uncertainty (in bits) in $X$. For any random variable $X$ with support on $\{\sigma_1, \ldots, \sigma_n\}$ such that $\Pr(\sigma_i) = p_i$, the Shannon entropy of $X$ is defined as
\begin{align*}
    H(X) = \sum_{i=1}^n p_i \log\left(\frac{1}{p_i}\right)
\end{align*}

\noindent The conditional entropy of a random variable $X$ given some information $Y$ about it, measures the expected amount of uncertainty left in $X$ after knowing $Y$. For two random variables $X, Y$, the conditional entropy of $X$ given $Y$ is defined as
    \begin{align*}
        H(X|Y) = H(XY) - H(Y) = \E_{y \sim Y}[H(X|Y=y)]
    \end{align*}

\noindent The mutual information of two random variables $X, Y$ measures the amount of information learnt about $X$ on learning $Y$ (or vice-versa). For two random variables $X, Y$, the mutual information between $X$ and $Y$ is defined as
    \begin{align*}
        I(X;Y) = H(X) + H(Y) - H(XY) = H(X) - H(X|Y) = H(Y) - H(Y|X)
    \end{align*}

\noindent The conditional mutual information of two random variables $X, Y$ conditioned on $Z$, measures the expected amount of information learnt about $X$ on learning $Y$ (or vice-versa), when you already know $Z$. For any random variables $X, Y, Z$, the mutual information between $X$ and $Y$ conditioned on $Z$ is defined as
    \begin{align*}
        I(X;Y|Z) = H(X|Z) - H(X|YZ) = H(Y|Z) - H(Y|XZ) = \E_{z \sim Z}[I(X;Y|Z = z)]
    \end{align*}

\noindent The mutual information and entropy satisfy the following properties:
\begin{itemize}
    \item Monotonicity: For any random variables $X, Y, Z$,
    \begin{align*}
        0 &\leq H(X|Y) \leq H(X) &\text{(Conditioning reduces entropy)} \\
        0 &\leq I(X;Y|Z) \leq \min \{H(X|Z), H(Y|Z)\} &\text{(Information learnt is at most the uncertainty)}
    \end{align*}
    \item Additivity: For any random variables $X_1, X_2, Y_1, Y_2, Z$,
    \begin{align*}
        H(X_1X_2|Z) &= H(X_1|Z) + H(X_2|Z) &\text{($X_1, X_2$ independent conditioned on $Z$)} \\
        I(X_1X_2; Y_1Y_2|Z) &= I(X_1; Y_1|Z) + I(X_2; Y_2|Z) &\text{($(X_1, Y_1), (X_2,Y_2)$ independent conditioned on $Z$)}
    \end{align*}
    \item Chain rule: For random variables $X_1, X_2, Y, Z$,
    \begin{align*}
        H(X_1X_2|Y) &= H(X_1|Y) + H(X_2|X_1Y) \\
        I(X_1X_2;Y|Z) &= I(X_1;Y|Z) + I(X_2;Y|X_1Z)
    \end{align*}
\end{itemize}

\subsubsection{Typicality and Conditional Typicality}
Given a random variable $X$, consider a sequence $x = (x_1, \ldots, x_n)$ containing $n$ outcomes from the support of $X$. The empirical distribution of $x$, denoted $P_{x}$, is the probability distribution of seeing a symbol $\sigma_i$ in the support of $X$ when a uniformly random index of $x$ is sampled, i.e.,
\begin{align*}
    P_{x}(\sigma_i) &= \Pr_{j \sim [n]}[x_j = \sigma_i]
\end{align*}
We can use this empirical distribution to classify the set of strongly-typical sequences $x$.
\begin{definition}[Strong Typicality]
    For a random variable $X$ with distribution $P$ and $\delta > 0$, a sequence $x = (x_1, \ldots, x_n)$ is called a strongly $\delta$-typical sequence if
    \begin{align*}
        \max_{\sigma_i \in \mathsf{supp}(X)}\left|P_{x}(\sigma_i) - P(\sigma_i)\right| \leq \delta
    \end{align*}
    and $P_{x}(\sigma_i) = 0$ if $P(\sigma_i) = 0$. The set $T_{\delta}$ denotes the set of strongly $\delta$-typical sequences $x$. Then using the law of large numbers,
    \begin{align*}
        \lim_{n \rightarrow \infty} \Pr_{x \sim X^n}[x \in T_{\delta}] = 1
    \end{align*}
\end{definition}

\begin{definition}[Joint Typicality]
    For random variables $X, Y$ with joint distribution $P$ and $\delta > 0$, a sequence $(x, y) = ((x_1, y_1), \ldots, (x_n, y_n))$ is called a jointly $\delta$-typical sequence if
    \begin{align*}
        \max_{(\alpha, \beta) \in \mathsf{supp}(X,Y)}\left|P_{(x, y)}(\alpha, \beta) - P(\alpha, \beta)\right| \leq \delta
    \end{align*}
    and $P_{(x, y)}(\alpha, \beta) = 0$ if $P(\alpha, \beta) = 0$.
\end{definition}

\begin{definition}[Conditional Typicality]
    For random variables $X, Y$ with joint distribution $P$, $\delta > \delta' > 0$, given a strongly $\delta'$-typical sequence $x = (x_1, \ldots, x_n)$ of symbols from $X$, a sequence $y = (y_1, \ldots, y_n)$ of symbols from $Y$ is called a conditionally $\delta$-typical sequence (conditioned on $x$) if $(x, y)$ is a jointly $\delta$-typical sequence with respect to $P$. The set $T_{\delta}^{x}$ denotes the set of conditional $\delta$-typical sequences conditioned on $x$.
\end{definition}

\begin{lemma}[Conditional Typicality Lemma]\label{lem:condtyp}
    For random variables $X, Y$ with joint distribution $P$, given $\delta > \delta' > 0$, let $x = (x_1, \ldots, x_n)$ be a strongly $\delta'$-typical sequence with respect to the marginal distribution on $X$. Then for a constant $c > 0$ and sufficiently large $n$,
    \begin{align*}
        (1-\delta) 2^{n(1-c\delta)H(Y|X)} \leq |T^{x}_{\delta}| \leq 2^{n(1+c\delta)H(Y|X)}
    \end{align*}
    If $y_i$ for the sequence $y = (y_1, \ldots, y_n)$ is sampled from the marginal distribution on $Y$ conditioned on $x_i$, then
    \begin{align*}
        \lim_{n \rightarrow \infty} \Pr_{y}[y \in T^{x}_{\delta}] = 1
    \end{align*}
\end{lemma}
\noindent See Section 13.9 of \cite{Wil13} for a proof of \cref{lem:condtyp}.

\subsection{Randomized Depth and Size Complexity of a Function}\label{sec:depthandsize}
In this section, we define the randomized tree depth and tree size complexity of a function.
\paragraph{Tree Depth.} For a function $f: \{0,1\}^k \rightarrow \{0,1\}$ and an input distribution $\mu$, the randomized query complexity of $f$ for success probability $\gamma$ is defined as:
\begin{align*}
    \mathsf{R}_{\gamma}(f) &= \min_{R: \mathsf{success}_f(R) \geq \gamma} \max_{x \in X} ~\max_{D \in \mathsf{supp}(R)} \mathsf{depth}(D,x) &\text{(Worst-case randomized complexity)}
\end{align*}
When the algorithm knows the input distribution $\mu$, we can also consider its average randomized query-complexity. In particular, we measure both the success and the query complexity in expectation over both the internal randomness and input distribution. This gives us the distributional randomized complexity of $f$ with success probability $\gamma$, with respect to distribution $\mu$:
\begin{align*}
    \overline{\mathsf{R}}_{\gamma}^{\mu}(f) &= \min_{R: \mathsf{success}^{\mu}_f(R) \geq \gamma} ~\E_{x \sim \mu, D \sim R}[\mathsf{depth}(D, x)] &\text{(Distributional randomized complexity)}
\end{align*}
We denote $\overline{\mathsf{depth}}^{\mu}(R) = \E_{x \sim \mu, D \sim R}[\mathsf{depth}(D, x)]$. When $\gamma = \frac{2}{3}$, we drop the subscript $\gamma$ for notational simplicity. Clearly, $\overline{\mathsf{R}}_{\gamma}^{\mu}(f) \leq \mathsf{R}_{\gamma}(f)$, and in general the distributional complexity can be asymptotically smaller than the worst-case complexity \cite{BB25}. However, when $\gamma = \frac{2}{3}$, the two quantities are asymptotically equal, when the former is maximized over input distributions.
\begin{lemma}[\cite{Ver98}]\label{lem:minmaxquery}
    For any function $f: \{0,1\}^k \rightarrow \{0,1\}$, we have the following:
    \begin{align*}
        \max_{\mu}\overline{\mathsf{R}}^{\mu}(f) = \Theta(\mathsf{R}(f))
    \end{align*}
\end{lemma}

\paragraph{Tree Size.} We define the size of a deterministic decision tree $D$, denoted $\mathsf{size}(D)$, as the number of leaves in the tree. We can use this to define the randomized tree size complexity of a function $f$, for success probability $\gamma$:
\begin{align*}
    \mathsf{RSize}_{\gamma}(f) &= \min_{R:\mathsf{success}_f(R) \geq \gamma} \max_{D \in \mathsf{supp}(R)} \mathsf{size}(D) &\text{(Worst-case size with worst-case success)} \\
    \overline{\mathsf{RSize}}_{\gamma}(f) &= \min_{R:\mathsf{success}_f(R) \geq \gamma} \E_{D \sim R} [\mathsf{size}(D)] &\text{(Expected size with worst-case success)}
\end{align*}
When the algorithm knows the input distribution $\mu$, we measure the success probability of the algorithm with respect to $\mu$, leading to the distributional randomized tree size complexity of $f$ with success probability $\gamma$:
\begin{align*}
    \mathsf{RSize}^{\mu}_{\gamma}(f) &= \min_{R:\mathsf{success}^{\mu}_f(R) \geq \gamma} \max_{D \in \mathsf{supp}(R)} \mathsf{size}(D) &\text{(Worst-case size with expected success)} \\
    \overline{\mathsf{RSize}}_{\gamma}^{\mu}(f) &= \min_{R:\mathsf{success}^{\mu}_f(R) \geq \gamma} \E_{D \sim R} [\mathsf{size}(D)] &\text{(Expected size with expected success)}
\end{align*}
It is also possible to relate the distributional complexity measures to the worst-case complexity measures, using minimax theorems.
\begin{lemma}[Minimax theorem for tree size]\label{lem:minmaxsize}
    For any function $f: \{0,1\}^k \rightarrow \{0,1\}$ and success probability $\gamma$, we have the following:
    \begin{align*}
        \max_{\mu} \mathsf{RSize}^{\mu}_{\gamma}(f) &= \mathsf{RSize}_{\gamma}(f) \\
        \max_{\mu} \overline{\mathsf{RSize}}_{\gamma}^{\mu}(f) &\leq \overline{\mathsf{RSize}}_{\gamma}(f)
    \end{align*}
\end{lemma}
Note that the measure $\mathsf{RSize}^{\mu}_{\gamma}(f)$ measures the worst-case of a tree in the support of the algorithm. As such, we can pick a decision tree in the support whose expected success (over the input distribution) is at least $\gamma$, and its size will be bounded by the worst-case size. So this measure is equal to $\mathsf{DSize}^{\mu}_{\gamma}(f)$, where we only consider deterministic trees. However, this is no longer true for randomized trees computing multiple instances of a function $f$, since different trees in the support could have high expected success for different copies of $f$ and we may not be able to pick a single tree from the support with high success on all instances.

\subsection{Information Complexity of a Function}
\label{sec:icinquery}
In the distributional setting, we can measure the cost of a decision tree in terms of the amount of information we learn about the input instead of the depth of its leaves.
Specifically, if $X \sim \mu$ and $L$ is a random variable that denotes the leaf of $T$ that corresponds to $X$, we define the \emph{information cost} of a deterministic tree $T$ to be the mutual information between $X$ and $L$:
\[
\mathsf{info}^\mu(T) = \rI(X; L).
\]
This amount also equals the entropy of the sequence of bits revealed by the queries of the deterministic tree $T$ during its execution. The information cost of a randomized decision tree $R$ is the expected information cost of a tree $T$ drawn from its distribution: $\mathsf{info}^\mu(R) = \E_{T \sim R}[ \mathsf{info}^\mu(T) ]$. This quantity also equals the expected (over the randomness of the tree) entropy of the sequence of bits revealed by the queries of the tree $R$.

The \emph{information complexity} of a Boolean function $f$ with respect to the distribution $\mu$ and success parameter $\gamma \in (\frac12, 1]$ is the minimal information cost of a randomized decision tree that computes $f$ with success probability at least $\gamma$ over inputs drawn from $\mu$.
\[
\mathsf{IC}^\mu_\gamma(f) = \min_{R : \sf success^\mu_f(R) \ge \gamma} \mathsf{info}^\mu(R).
\]
We can also define a distribution-free information complexity measure of Boolean functions by taking the supremum over all distributions on its domain: ${\sf IC}_\gamma(f) = \sup_{\mu} {\sf IC}^\mu_\gamma(f)$.
Furthermore, for notational simplicity, when $\gamma = \frac23$, we simply write ${\sf IC}(f) = {\sf IC}_{\frac23}(f)$ and ${\sf IC}^\mu(f) = {\sf IC}^{\mu}_{\frac 23}(f)$.

In this section, we show some basic properties of the information complexity of a function $f$. We first show that the information complexity is upper bounded by the logarithm of expected tree size with success probability $\gamma$ with respect to distribution $\mu$.
\begin{lemma}\label{lem:treesizebound}
    For any function $f: \{0,1\}^k \rightarrow \{0,1\}$, distribution $\mu$ on $\{0,1\}^k$ and score $\gamma$,
    \begin{align*}
        \IC^\mu_\gamma(f) \leq \log\left(\overline{\mathsf{RSize}}^\mu_\gamma(f)\right)
    \end{align*}
\end{lemma}
\begin{proof}
    Let $R$ be any randomized decision tree computing the function $f$ with $\mathsf{success}^{\mu}_f(R) \geq \gamma$.
    \begin{align*}
        \info^{\mu}(R) &= \E_{D \sim R} [\mathsf{info^{\mu}(D)}] = \E_{D \sim R} [\rI(X; L)] \\
        &\leq \E_{D \sim R} [\log(\mathsf{size}(D))] \leq \log\left(\E_{D \sim R} [(\mathsf{size}(D))]\right) &\text{(Jensen's inequality)}
    \end{align*}
    Since the choice of $R$ was arbitrary, the claim follows.
\end{proof}
We now show that the information complexity of a function is continuous in the success probability parameter $\gamma$ in the interval $\gamma \in \left(\frac{1}{2}, 1\right)$.

\begin{lemma}
    For any function $f: \{0,1\}^k \rightarrow \{0,1\}$, and distribution $\mu$ on $\{0,1\}^k$, the information complexity $\IC^{\mu}_{\gamma}(f)$ is a convex function of the score $\gamma \in \left[\frac{1}{2}, 1\right]$.
\end{lemma}
\begin{proof}
    For any scores $\gamma_1, \gamma_2 \in \left[\frac{1}{2}, 1\right]$ and $\alpha \in (0,1)$, consider randomized decision trees $R_1, R_2$ with score at least $\gamma_1, \gamma_2$ respectively such that $\info^{\mu}(R_i) = \IC^{\mu}_{\gamma_i}(f)$ for $i \in \{1,2\}$. Set $\gamma = \alpha \gamma_1 + (1-\alpha)\gamma_2$. Then consider the decision tree $R$ which runs $R_1$ with probability $\alpha$ and $R_2$ with probability $1-\alpha$. Then $R$ has score at least $\gamma$, and
    \begin{align*}
        \info^{\mu}(R) = \E_{D \sim R} [\mathsf{info^{\mu}(D)}] = \alpha \E_{D \sim R_1} [\mathsf{info^{\mu}(D)}] + (1-\alpha) \E_{D \sim R_2} [\mathsf{info^{\mu}(D)}]
    \end{align*}
    Therefore, we get that
    \begin{align*}
        \IC^{\mu}_{\gamma}(f) \leq \alpha \IC^{\mu}_{\gamma_1}(f) + (1-\alpha) \IC^{\mu}_{\gamma_2}(f) &\qedhere
    \end{align*}
\end{proof}
\begin{corollary}\label{cor:conticscore}
    For any function $f: \{0,1\}^k \rightarrow \{0,1\}$, and distribution $\mu$ on $\{0,1\}^k$, the information complexity $\IC^{\mu}_{\gamma}(f)$ is a continuous function of the score $\gamma \in \left(\frac{1}{2}, 1\right)$.
\end{corollary}

We now define a cost measure on the leaves of a tree using its surprisal. We then show that the information cost of a tree is the same as its expected surprisal cost.

\paragraph{Surprisal based Cost Measure.} Given an input distribution $\mu$, and a leaf $\ell$ of a deterministic decision tree $D$, we define the surprisal cost of $\ell$ with respect to $\mu$ as:
\begin{align*}
    \mathsf{surprisal}^{\mu}(\ell) = -\log\left(\Pr_{\mu}[\ell]\right)
\end{align*}
where $\Pr_{\mu}[\ell]$ is the probability of sampling an input from $\mu$ that leads to the leaf $\ell$. We now show some properties of the surprisal cost.
\begin{itemize}
    \item Empty leaf (no queries) has surprisal cost $0$ because every input is consistent with it.
    \item Monotonicity: Consider an internal node $v$ of $D$ on the path from root to leaf $\ell$. Then $\Pr_{\mu}[v] \geq \Pr_{\mu}[\ell]$, therefore $\mathsf{surprisal}^{\mu}(v) \leq \mathsf{surprisal}^{\mu}(\ell)$. So the surprisal cost increases as we traverse down a path of the tree.
    \item Additivity: Suppose $\ell = \ell_1 \times \ldots \times \ell_n$ is a leaf with queries on $(\{0,1\}^{m})^{n}$ and $\ell_i$ consists of the queries on the $i^{th}$ input from $\{0,1\}^m$ for $i \in [n]$, then
    \begin{align*}
        \mathsf{surprisal}^{\mu^n}(\ell) = -\log\left(\Pr_{\mu^n}[\ell_1 \times \ldots \times \ell_n]\right) = -\sum_{i=1}^n \log\left(\Pr_{\mu}[\ell_i]\right) = \sum_{i=1}^n \mathsf{surprisal}^{\mu}(\ell_i)
    \end{align*}
\end{itemize}
The surprisal cost of a deterministic tree $D$ is the expected surprisal cost of the leaves of the tree with respect to the distribution $\mu$.
\begin{align*}
    \overline{\mathsf{surprisal}}^{\mu}(D) = \E_{\ell \sim L(\mu, D)}[\mathsf{surprisal}^{\mu}(\ell)]
\end{align*}
The surprisal cost of a randomized tree $R$ is then the expected surprisal cost of the deterministic trees $D$ in the support of $R$ with respect to the distribution $\mu$.
\begin{align*}
    \overline{\mathsf{surprisal}}^{\mu}(R) = \E_{D \sim R}[\overline{\mathsf{surprisal}}^{\mu}(D)]
\end{align*}
 We now show that the information cost of a deterministic decision tree is equal to the expected surprisal of its leaves with respect to the input distribution $\mu$.
\begin{proposition}
    For any deterministic decision tree $D$, the information cost of $D$ is equal to its expected surprisal cost with respect to the input distribution $\mu$:
    \begin{align*}
        \mathsf{info}^{\mu}(D) = \overline{\mathsf{surprisal}}^{\mu}_f(D)
    \end{align*}
\end{proposition}
\begin{proof}
    Let $L$ be the random variable corresponding to the leaf $\ell$ of $D$. The claim follows because
    \begin{align*}
        \mathsf{info}^{\mu}(D) = \rI(X; L) = H(L) - H(L|X) = H(L) = \E_{\ell \sim L(\mu, D)}[\mathsf{surprisal}^{\mu}(\ell)] &\qedhere
    \end{align*}
\end{proof}

\begin{lemma}
    For any function $f: \{0,1\}^k \rightarrow \{0,1\}$, input distribution $\mu$ and success probability $\gamma$, we have that
    \begin{align*}
        \IC^{\mu}_{\gamma}(f) = \min_{R: \mathsf{success}^{\mu}_f(R) \geq \gamma} \overline{\mathsf{surprisal}}^{\mu}_f(R)
    \end{align*}
\end{lemma}
\begin{proof}
    This follows because
    \begin{align*}
        \mathsf{IC}^\mu_\gamma(f) &= \min_{R : \mathsf{success}^\mu_f(R) \ge \gamma} \mathsf{info}^\mu(R) = \min_{R : \mathsf{success}^\mu_f(R) \ge \gamma} \E_{D\sim R}[\mathsf{info}^\mu(D)] \\ 
        &= \min_{R : \mathsf{success}^\mu_f(R) \ge \gamma} \E_{D\sim R}[\overline{\mathsf{surprisal}}^\mu(D)] = \min_{R: \mathsf{success}^{\mu}_f(R) \geq \gamma} \overline{\mathsf{surprisal}}^{\mu}_f(R) \qedhere
    \end{align*}
\end{proof}

\paragraph{Success-conditioned Information Complexity.} We now define a success-conditioned version of information complexity. This was defined in \cite{BB25} for randomized query complexity, i.e., when we measure the depth of the tree. The success-conditioned information cost measures the expected surprisal cost of an algorithm $R$ conditioned on the algorithm giving the correct answer. This measure rules out algorithms which achieve low expected surprisal by failing most of the time and succeeding with exponentially small probability by querying the entire input. For any leaf $\ell$, we use $\mathsf{success}^{\mu}_f(\ell)$ to denote the expected (over input distribution) success probability conditioned on reaching leaf $\ell$. The success-conditioned information cost of an algorithm $R$ is defined as
\begin{align*}
    \mathsf{sinfo}^{\mu}_f(R) = \frac{\E_{\ell \sim L(\mu, R)}[\mathsf{success}^{\mu}_f(\ell)\cdot \mathsf{surprisal}^{\mu}(\ell)]}{\mathsf{success}^{\mu}_f(R)}
\end{align*}
The success-conditioned information complexity of $f$ with success probability $\gamma$ is defined as:
\begin{align*}
    \mathsf{sIC}_{\gamma}^{\mu}(f) &= \min_{R: \mathsf{success}^{\mu}_f(R) \geq \gamma} \mathsf{sinfo}^{\mu}_f(R) &\text{(Success-conditioned information complexity)}
\end{align*}
We can now also define the maximum success-conditioned information complexity of $f$ for success probability $\gamma$.
\begin{align*}
    \mathsf{sIC}_{\gamma}(f) &:= \max_{\mu} \mathsf{sIC}_{\gamma}^{\mu}(f) &\text{(Maximum success-conditioned information complexity)}
\end{align*}
We now relate the success-conditioned information complexity to information complexity, a similar relation was shown for depth complexity in \cite{BB25}.
\begin{proposition}\label{prop:srcvsrc}
    For any function $f: \{0,1\}^k \rightarrow \{0,1\}$, input distribution $\mu$ and success probability $\gamma$,
    \begin{align*}
        \frac{1}{2}\IC^{\mu}_{\gamma}(f) \leq \mathsf{sIC}_{\gamma}^{\mu}(f)
    \end{align*}
\end{proposition}
\begin{proof}
    Since $\mathsf{success}_f^{\mu}(\ell) \geq \frac{1}{2}$ for any leaf $\ell$ and $\mathsf{success}^{\mu}_f(R) \leq 1$ for any algorithm $R$,
    \begin{align*}
        \mathsf{sIC}_{\gamma}^{\mu}(f) \geq \min_{R: \mathsf{success}^{\mu}_f(R) \geq \gamma} \E_{\ell \sim L(\mu, R)}\left[\frac{\mathsf{surprisal}^{\mu}(\ell)}{2}\right] = \frac{1}{2} \IC^{\mu}_{\gamma}(f) &\qedhere
    \end{align*}
\end{proof}
Now we show that $\mathsf{sIC}^{\mu}_{\gamma}(f)$ is upper bounded by the logarithm of the expected decision-tree size (up to an additive loss).
\begin{lemma}\label{lem:sicub}
    For any function $f: \{0,1\}^k \rightarrow \{0,1\}$, distribution $\mu$ on $\{0,1\}^k$ and success probability $\gamma$,
    \begin{align*}
        \mathsf{sIC}^{\mu}_{\gamma}(f) &\leq 1+ \log\left(\overline{\mathsf{RSize}}_{\gamma}^{\mu}(f)\right) \\
        \mathsf{sIC}^{\mu^n}_{\gamma^n}(f^n) &\leq n+ \log\left(\overline{\mathsf{RSize}}_{\gamma^n}^{\mu^n}(f^n)\right)
    \end{align*}
\end{lemma}
\begin{proof}
    Let $R$ be a decision tree computing $f$ such that $\mathsf{success}^{\mu}_f(R) \geq \gamma$. Then
    \begin{align*}
        \mathsf{sinfo}^{\mu}_f(R) &= \frac{\E_{D \sim R}\E_{\ell \sim L(\mu, D)}[\mathsf{success}^{\mu}_f(\ell) \cdot \mathsf{surprisal}^{\mu}(\ell)]}{\mathsf{success}^{\mu}_f(R)} \\
        &= \frac{1}{\mathsf{success}^{\mu}_f(R)}\sum_{D \in R}\left[\Pr_R[D]\sum_{\ell \in  \mathsf{supp}(D)}\left[\Pr_{\mu}[\ell] \cdot \mathsf{success}^{\mu}_f(\ell) \cdot \log\left(\frac{1}{\Pr_{\mu}[\ell]}\right)\right]\right] \\ 
        &\leq \log\left( \frac{\sum_{D \in R}\left[\Pr_R[D]\sum_{\ell \in  \mathsf{supp}(D)}\left[\Pr_{\mu}[\ell] \cdot \mathsf{success}^{\mu}_f(\ell) \cdot \frac{1}{\Pr_{\mu}[\ell]}\right]\right]}{\mathsf{success}^{\mu}_f(R)}\right) \\
        &\leq \log\left(\frac{\sum_{D \in R}\Pr_R[D]\mathsf{size}(D)}{\mathsf{success}^{\mu}_f(R)}\right) \\
        &\leq \log\left(\frac{\E_{D \sim R}[\mathsf{size}(D)]}{\gamma}\right) \\
        &\leq 1 + \log\left(\E_{D \sim R}[\mathsf{size}(D)]\right)
    \end{align*}
    Here the third inequality follows from Jensen's inequality. Minimizing over choices of $R$ satisfying $\mathsf{success}^{\mu}_f(R) \geq \gamma$ on both sides above gives us that
    \begin{align*}
        \mathsf{sIC}^{\mu}_{\gamma}(f) &\leq 1+\log\left(\overline{\mathsf{RSize}}_{\gamma}^{\mu}(f)\right)
    \end{align*}
    A similar proof can be done for $f^n$, and we get the bound from the statement using the fact that $\gamma^n \geq \frac{1}{2^n}$.
\end{proof}

\paragraph{Computing $n$ copies of $f$.} When defining the complexity of computing $n$ copies of $f$, we will consider the two notions of success probability -- per-copy success and overall success. The information complexity for $f^n$ with respect to distribution $\mu^n$ can then be defined as
\begin{align*}
    \IC_{\gamma^{\otimes n}}^{\mu^n}(f^n) &= \min_{R:\forall i~\mathsf{success}^{\mu}_{f_i}(R) \geq \gamma} \mathsf{info}^{\mu^n}(R) &\text{(Information complexity with per-copy success)} \\
    \IC_{\gamma^{n}}^{\mu^n}(f^n) &= \min_{R:\mathsf{success}^{\mu^n}_{f^n}(R) \geq \gamma^n} \mathsf{info}^{\mu^n}(R) &\text{(Information complexity with overall success)} \\
    \mathsf{sIC}_{\gamma^{n}}^{\mu^n}(f^n) &= \min_{R:\mathsf{success}^{\mu^n}_{f^n}(R) \geq \gamma^n} \mathsf{sinfo}^{\mu^n}_{f^n}(R) &\text{(Success-conditioned information with overall success)}
\end{align*}
We can define analogous versions of tree size (and depth) for computing $f^n$. When considering per-copy success probability, we use the subscript $\gamma^{\otimes n}$, and when considering overall success probability, we use the subscript $\gamma^n$.

\section{Information is Amortized Size Complexity}\label{sec:amortization}
In this section, we show that information complexity is equal to amortized (per-copy) logarithm of tree size (\cref{thm:amortization}). This result is analogous to the ``information equals amortized communication'' result in communication complexity.

\subsection{Proof Overview} 

The proof of \cref{thm:amortization} regarding the amortized size complexity of functions uses the per-leaf surprisal based definition of information and is obtained by establishing inequalities in both directions. 

To prove that the amortized size complexity is bounded above by information complexity, we start by considering the tree $B$  computing $f^n$ on $\mu^n$ that is obtained by simply simulating $n$ independent copies of a randomized tree $A$ computing $f$ on $\mu$. This tree $B$ may have size much larger than the desired bound $2^{n \cdot \info^{\mu}(A)}$. But when $n$ is large enough, the law of large numbers lets us categorize all inputs $x \sim \mu^n$ into ``typical'' and ``atypical'' inputs. In turn, this lets us partition the set of leaves of $B$ into typical and atypical leaves. The upper bound is completed by noticing that the set of atypical leaves has measure that tends to $0$ when $n \to \infty$ and that when we truncate $B$ to remove these leaves, we do obtain a tree $B'$ with the desired size complexity.

The lower bound of \cref{thm:amortization} is established using the additivity of information complexity, as established in \cref{lem:additivity}. The proof of this lemma is obtained by noting that for any algorithm $B$ computing $f^n$ on the input distribution $\mu^n$, a smaller algorithm $A$ can be constructed to compute $f$ on the input distribution $\mu$, which places the input at a random location of $B$, and samples the remaining inputs for $B$ itself. Then if $B$ makes a query to a fake input, it learns no information about the real input, since the $n$ inputs are independent conditioned on the query responses seen by $B$ so far. So the average information complexity of $A$ is exactly $1/n$ times the information complexity of $B$.

\subsection{Upper bound on Amortized Size Complexity} 

We start by showing that when computing $n$ copies of $f$ (as $n \rightarrow \infty$) on the distribution $\mu^n$, we can compress the logarithm of worst-case tree size with expected success $\gamma$ for each copy of $f$, to $n$ times the information complexity of a single copy of $f$.
\begin{lemma}\label{lem:sizecomplim}
    For any function $f: \{0,1\}^k \rightarrow \{0,1\}$, distribution $\mu$ on $\{0,1\}^k$ and success probability $\gamma < 1$,
    \begin{align*}
        \lim_{n \rightarrow \infty} \frac{\log\left(\mathsf{RSize}^{\mu^n}_{\gamma^{\otimes n}}(f^n)\right)}{n} \leq \IC^\mu_\gamma(f)
    \end{align*}
\end{lemma}
\begin{proof}
     For arbitrary $\epsilon \in (0, 1-\gamma)$, let $A$ be a randomized decision tree computing $f$ on input distribution $\mu$ with success probability $\gamma + \epsilon$. Denote the randomness of $A$ by the random variable $R$, and the sampled leaf by $L$. Then we know that $\mathsf{info}^{\mu}(A) = \rI(L;X|R) = H(L|R) - H(L|XR) = H(L|R)$. Define the tree $A^n$ which runs $n$ independent copies of $A$. Clearly, $A^n$ computes $f^n$ on the distribution $\mu^n$, with per-copy success probability $\gamma + \epsilon$. Let $0 < \delta' < \delta$. Then we define a truncated randomized decision tree $A^n_{\delta, \delta'}$ which is constructed as follows:
    \begin{enumerate}
        \item Sample randomness $r = r_1, \ldots, r_n$ according to the randomized tree $A^n$. Denote the set of strongly $\delta'$-typical sequences by $T_{\delta'}$. If $r \notin T_{\delta'}$, the corresponding tree in the support of $A^n_{\delta, \delta'}$ makes no queries and outputs $0^n$. Otherwise, let the corresponding decision tree in support of $A^n$ be $D_r$.
        \item Let $\ell = \ell_1, \ldots, \ell_n$ be sampled according to the distribution $L(\mu^n, D_r)$. Denote the set of conditional $\delta$-typical sequences by $T^r_{\delta}$. If $\ell \notin T^r_{\delta}$, remove $\ell$ from $D_r$. Put the resulting tree $D_r^{\mathsf{typ}}$ in the support of $A^n_{\delta, \delta'}$.
    \end{enumerate}
    We know from the law of large numbers that $\lim_{n \rightarrow \infty}\Pr_{r \sim A^n}[r \in T_{\delta'}] = 1$. Therefore, for every $\epsilon, \delta' > 0$, there exists sufficiently large $n_0 \in \mathbb{N}$ such that for all $n \geq n_0$,
    \begin{align*}
        \Pr_{r \sim A^{n}}[r \in T_{\delta'}] \geq 1 - \frac{\epsilon}{2}
    \end{align*}
    Further, conditioned on the fact that $r \in T_{\delta'}$, we know from the conditional typicality lemma (\cref{lem:condtyp}) that $\lim_{n \rightarrow \infty}\Pr_{\ell \sim D_r}[\ell \in T^r_{\delta}] = 1$. Therefore, for every $\epsilon > 0$ and $\delta > \delta' > 0$, and $r \in T_{\delta'}$ there exists sufficiently large $n_1 \in \mathbb{N}$ such that for all $n \geq n_1$,
    \begin{align*}
        \Pr_{\ell \sim D_r}[\ell \in T^r_{\delta}] \geq 1 - \frac{\epsilon}{2}
    \end{align*}
    In particular, this implies that for sufficiently large $n \geq \max\{n_0, n_1\}$, $A^n_{\delta, \delta'}$ computes $f^n$ on the distribution $\mu^n$, with per-copy success probability at least $\gamma + \epsilon - \frac{\epsilon}{2} - \frac{\epsilon}{2} = \gamma$. Moreover, conditioned on the fact that $r \in T_{\delta'}$, for sufficiently large $n$ and a constant $c > 0$,
    \begin{align*}
        &|T^r_{\delta}| \leq 2^{n(1+c\delta)H(L|R)} = 2^{n(1+c\delta)\mathsf{info}^{\mu}(A)} \\
        \implies &\mathsf{size}(D^{\mathsf{typ}}_r) \leq 2^{n(1+c\delta)\mathsf{info}^{\mu}(A)}
    \end{align*}
    In particular,
    \begin{align*}
        \lim_{\delta \rightarrow 0}\lim_{n \rightarrow \infty}\frac{\log\left(\max_{D \in \mathsf{supp}(A^n_{\delta, \delta'})}[\mathsf{size}(D)]\right)}{n} \leq \lim_{\delta \rightarrow 0}\lim_{n \rightarrow \infty}\frac{\log\left(2^{n(1+c\delta)\mathsf{info}^{\mu}(A)}\right)}{n}  \leq \mathsf{info}^{\mu}(A)
    \end{align*}
    Therefore, we get that for every $\epsilon \in (0, 1-\gamma)$,
    \begin{align*}
        \lim_{n \rightarrow \infty} \frac{\log\left(\mathsf{RSize}^{\mu^n}_{\gamma^{\otimes n}}(f^n)\right)}{n} \leq \IC^\mu_{\gamma+\epsilon}(f)
    \end{align*}
    Using continuity of information complexity (\cref{cor:conticscore}) in the success probability, we get that
    \begin{align*}
        \lim_{n \rightarrow \infty} \frac{\log\left(\mathsf{RSize}^{\mu^n}_{\gamma^{\otimes n}}(f^n)\right)}{n} \leq \lim_{\epsilon \rightarrow 0} \IC^\mu_{\gamma+\epsilon}(f) = \IC^\mu_\gamma(f) &\qedhere
    \end{align*}
\end{proof}

\subsection{Lower bound on Amortized Size Complexity} 

We now want to show that this is indeed the smallest possible tree size in the amortized setting. We establish this by showing that the information complexity $\IC$ satisfies an additivity property. That is, the information complexity for computing $n$ copies of a function $f$ with respect to $\mu^n$ and per-copy success probability $\gamma$ is equal to $n$ times the information complexity for computing the function $f$ with respect to $\mu$ with success probability $\gamma$.

\begin{lemma}\label{lem:qicdsub}
    Let $T_1$ and $T_2$ be deterministic decision trees running on input distributions $\mu_1$ and $\mu_2$ on $\{0,1\}^k$ respectively. Let $T$ be a decision tree that runs $T_1$ and $T_2$ on inputs sampled from $\mu = \mu_1 \otimes \mu_2$. Then
    \begin{align*}
        \info^\mu(T) = \info^{\mu_1}(T_1) + \info^{\mu_2}(T_2)
    \end{align*}
\end{lemma}
\begin{proof}
We compute the information cost of $T$.
    \begin{align*}
        \info^\mu(T) &= \rI(L;\rX_1,\rX_2) \\
        &= \rI(L_{T_1},L_{T_2};\rX_1,\rX_2) \\
        &= \rI(L_{T_1};\rX_1) + \rI(L_{T_2};\rX_2) &\text{(Because $(L_{T_1}, \rX_1)$ is independent of $(L_{T_2}, \rX_2)$)} \\
        &= \info^{\mu_1}(T_1) + \info^{\mu_2}(T_2) &\qedhere
    \end{align*}
\end{proof}

\begin{lemma}\label{lem:qicdslb}
    Suppose $T$ is a deterministic decision tree running on input sampled from the distribution $\mu = \mu_1 \otimes \mu_2$ on $\{0,1\}^k \times \{0,1\}^k$. Define the randomized decision tree $T_1(\cdot)$ which runs $T(\cdot, \mu_2)$ on inputs sampled from $\mu_1$ and $T_2(\cdot)$ which runs $T(\mu_1, \cdot)$ on inputs sampled from $\mu_2$. Then
    \begin{align*}
        \info^{\mu_1}(T_1) + \info^{\mu_2}(T_2) = \info^\mu(T)
    \end{align*}
\end{lemma}
\begin{proof}
    Suppose $T$ makes $t$ queries in the worst case. By definition, $T_1$ samples the inputs from $\mu_2$ to feed into $T$ using its own randomness, and similarly for $T_2$. Let $l_i$ denote the location of the $i^{th}$ query of $T$, and $q_i$ denote the query response for the $i^{th}$ query. Let $b_{i}$ in $l_i$ denote whether the query is made to $X_1$ or $X_2$. Note that if $T$ has already reached a leaf before the $i^{th}$ query, then $l_i, q_i, b_i$ will all be the empty string. Let $p_i$ denote the path on the decision tree up to the $i^{th}$ query, i.e., $p_i = l_1, q_1,\ldots, l_{i-1}, q_{i-1}, l_{i}, q_{i}$. Let $(p_{i-1}, l_i):b_i$ denote that the choice of $p_{i-1}, l_i$ is consistent with $b_i$. Then
    \begin{align*}
        \info^\mu(T) &= \rI(L; \rX_1,\rX_2) \\
        &= \sum_{i=1}^t \rI(q_i; \rX_1,\rX_2|l_1, q_1,\ldots, l_{i-1}, q_{i-1}, l_{i}) \\
        &= \sum_{i=1}^t \E_{(p_{i-1}, l_i)}[\rI(q_i; \rX_1, \rX_2|p_{i-1}, l_{i})] \\
        &= \sum_{i=1}^t \E_{b_i}\left[\E_{(p_{i-1}, l_i):b_i}[\rI(q_i; \rX_1, \rX_2|p_{i-1}, l_{i}, b_i)]\right] \\
        &= \sum_{i=1}^t \E_{b_i}\left[\E_{(p_{i-1}, l_i):b_i}[\rI(q_i; \rX_{b_i}|p_{i-1}, l_{i}, b_i)]\right]
    \end{align*}
    Here the last equality follows from \cref{fact:indepinputs} which says that conditioned on the partial assignment $p_{i-1}$, $\rX_1$ and $\rX_2$ are independent. So if $l_i$ is a query to $\rX_1$, $q_{i}$ is independent of $\rX_2$. Let us now look at the information cost of the trees $T_1$ and $T_2$.
    \begin{align*}
        \info^{\mu_1}(T_1) &= \rI(L;\rX_1|\rX_2) \\
        &= \sum_{i=1}^t \rI(q_{i};\rX_1|l_1, q_1,\ldots, l_{i-1}, q_{i-1}, l_{i},\rX_2) \\
        &= \sum_{i=1}^t \E_{p_{i-1}, l_i} [\rI(q_{i};\rX_1|p_{i-1}, l_i, \rX_2)] \\
        &= \sum_{i=1}^t \E_{b_i}\left[\E_{(p_{i-1}, l_i):b_i} [\rI(q_{i};\rX_1|p_{i-1}, l_i, b_i, \rX_2)]\right] \\
        &= \sum_{i=1}^t \left(\Pr[b_i = 1] \cdot \E_{(p_{i-1}, l_i)} [\rI(q_{i};\rX_1|p_{i-1}, l_i, b_i = 1)]\right)
    \end{align*}
    Here the last equality again follows from \cref{fact:indepinputs}, because conditioned on $b_i = 2$, $q_i$ is independent of $X_1$, and conditioned on $b_i = 1$, $X_2$ is independent of both $q_i$ and $X_1$.
    Similarly,
    \begin{align*}
        \info^{\mu_2}(T_2) = \rI(L;\rX_2|\rX_1) = \sum_{i=1}^t \left(\Pr[b_i = 2] \cdot \E_{(p_{i-1}, l_i)} [\rI(q_{i};\rX_2|p_{i-1}, l_i, b_i = 2)]\right)
    \end{align*}
    Therefore,
    \begin{align*}
        \info^{\mu_1}(T_1) + \info^{\mu_2}(T_2) = \info^{\mu}(T) &\qedhere
    \end{align*}
\end{proof}

\additivity*
\begin{proof}
    We first show that $\IC_{\gamma^{\otimes n}}^{\mu^n}(f^n) \leq n\IC_{\gamma}^{\mu}(f)$. Suppose $A$ is a randomized decision tree computing $f$ on distribution $\mu$, with success probability at least $\gamma$. Then define the decision tree $B$ which runs $n$ independent copies of $A$ on inputs sampled from $\mu^n$. Then clearly per-copy success probability of $B$ is $\gamma$. Moreover,
    \begin{align*}
        \info^{\mu^n}(B) &= \E_{r_1, \ldots, r_n} [\info^{\mu^n}(A_{r_1}\ldots A_{r_n})] \\
        &= \sum_{i=1}^n \E_{r_i}[\info^{\mu}(A_{r_i})] &\text{(From \cref{lem:qicdsub})}  \\
        &= n\cdot\info^{\mu}(A)
    \end{align*}
    Since $A$ is an arbitrary algorithm, we get that $\IC_{\gamma^{\otimes n}}^{\mu^n}(f^n) \leq n\IC_{\gamma}^{\mu}(f)$.\\
    \\
    \noindent For the other direction, suppose $B$ is a randomized decision tree computing $f^n$ on inputs sampled from $\mu^n$, with per-copy success probability at least $\gamma$. Let $A_i$ be the algorithm that runs $B$ on inputs sampled from $\mu$ in the $i^{th}$ location, and samples all the remaining inputs from $\mu$ itself. Then the success probability of $A_i$ is at least $\gamma$. When the deterministic tree in the support of $B$ is $D$, we call the corresponding randomized algorithm of $A_i$ as $D_i$. Now we analyze the information cost,
    \begin{align*}
        \info^{\mu^n}(B) &= \E_{D \sim B} \left[\info^{\mu^n}(D)\right] \\
        &= \E_{D \sim B} \left[\sum_{i=1}^n \info^{\mu}(D_i)\right] &\text{(From \cref{lem:qicdslb})} \\
        &= \sum_{i=1}^n \E_{D \sim B} [\info^{\mu}(D_i)] \\
        &= \sum_{i=1}^n \info^{\mu}(A_i)
    \end{align*}
    Now let $A$ be the algorithm that samples $i$ uniformly at random from $[n]$ and runs $A_i$. Then the success probability of $A$ is at least $\gamma$, and $\info^{\mu}(A) = \E_{i \sim [n]}[\info^{\mu}(A_i)] = \frac{\info^{\mu^n}(B)}{n}$. Since $B$ is an arbitrary algorithm, we get that $\IC_{\gamma^{\otimes n}}^{\mu^n}(f^n) \geq n\IC_{\gamma}^{\mu}(f)$.
\end{proof}

\noindent Using the additivity property, we get that information complexity equals amortized (per-copy) logarithm of tree size.

\amortization*
\begin{proof}
    The claim follows because
    \begin{align*}
        \IC^\mu_\gamma(f) \leq \lim_{n \rightarrow \infty} \frac{\IC_{\gamma^{\otimes n}}^{\mu^n}(f^n)}{n} \leq \lim_{n \rightarrow \infty} \frac{\log\left(\mathsf{RSize}^{\mu^n}_{\gamma^{\otimes n}}(f^n)\right)}{n} \leq \IC^\mu_\gamma(f) &\qedhere
    \end{align*}
\end{proof}

\section{Information Allows for Tree Size Compression}\label{sec:compression}
In this section, we show that even for computing a single copy of $f$ with constant success probability, we can compress the logarithm of the worst-case tree size to the information complexity (up to a constant factor), with a small loss in the success probability.

\subsection{Proof Overview} 

The per-leaf definition of information complexity leads to a particularly concise proof of \cref{thm:compression} on the compression of tree size complexity.
First, we show that eliminating trees with large information cost from the support of a randomized algorithm does not increase the algorithm's error very much.  Then, in the main part of the argument, we show that we can truncate the remaining trees to eliminate all of the leaves that are reached with very small probability (and thus have large surprisal values). This gives us trees with small size and only increases the error of the overall algorithm by another small amount.

\subsection{Proof of \cref{thm:compression}}

\compression*
\begin{proof}
    Let $A$ be a randomized decision tree computing $f$ on input distribution $\mu$ with success probability $\gamma$. For any fixed randomness $r$, let the corresponding tree in the support of $A$ be $A_r$. Let $\rI(X; L|R = r) = i_r$. If $i_r > 2\info^{\mu}(A)/\epsilon$, we remove $A_r$ from the support of $A$ and output $0$. This happens with probability at most $\epsilon/2$ over the choice of randomness $r$. Otherwise, we truncate $A_r$, to remove all leaves $\ell$ which have probability $p_\ell \leq 2^{-2i_r/\epsilon}$. Denote the set of remaining leaves by $P_r$. Then we have that,
    \begin{align*}
        i_r = \sum_{\ell \in A_r} p_\ell \log \left(\frac{1}{p_\ell}\right) \geq \sum_{\ell \notin P_r} p_\ell \log \left(\frac{1}{p_\ell}\right) \geq \frac{2i_r}{\epsilon} \sum_{\ell \notin P_r} p_\ell
    \end{align*}
    Therefore, $\sum_{\ell \notin P_r} p_\ell \leq \epsilon/2$. Moreover, we have that every $\ell \in P_r$ has $p_\ell \geq 2^{-2i_r/\epsilon}$, so
    \begin{align*}
        |P_r| \leq 2^{2i_r/\epsilon} \leq 2^{(2/\epsilon)^2 \info^{\mu}(A)}
    \end{align*}
    Therefore, for the resulting tree $A'$, we have that
    \begin{align*}
        \max_{D \in \mathsf{supp}(A')}[\mathsf{size}(D)] &\leq 2^{(2/\epsilon)^2 \info^{\mu}(A)}
    \end{align*}
    Moreover, $A'$ computes $f$ on $\mu$ with success probability at least $\gamma - \epsilon/2 - \epsilon/2 = \gamma - \epsilon$. Since the choice of $A$ was arbitrary, the claim follows.
\end{proof}

\subsection{Direct Sum Theorem for Tree Size}

As a corollary of the compression theorem, we get a direct sum theorem for randomized tree size. Direct sum theorems for query complexity have previously been studied in \cite{JKS10, BB19, BKST24} and for communication complexity in \cite{BBCR13}.

\begin{corollary}\label{cor:sizeds}
    For any function $f: \{0,1\}^k \rightarrow \{0,1\}$, distribution $\mu$ on $\{0,1\}^k$, success probability $\gamma$ and $0 < \epsilon \leq \gamma - \frac{1}{2}$,
    \begin{align*}
        n\log\left(\overline{\mathsf{RSize}}_{\gamma}^{\mu}(f)\right) \geq \log\left(\overline{\mathsf{RSize}}_{\gamma^{\otimes n}}^{\mu^n}(f^n)\right) \geq \frac{n\epsilon^2}{4}\log\left(\mathsf{RSize}_{\gamma-\epsilon}^{\mu}(f)\right)
    \end{align*}
\end{corollary}
\begin{proof}
    To see the first inequality, let $A$ be a randomized tree computing $f$ on distribution $\mu$ with success probability at least $\gamma$. Then we define the tree $B$ for computing $f^n$ on $\mu^n$, which runs $n$ independent copies of $A$. Then the per-copy success probability of $B$ is at least $\gamma$, and the expected size of $B$ is $\E_{r_1, \ldots, r_n}[\prod_{i=1}^n \mathsf{size}(A_{r_i})] = \prod_{i=1}^n \E_{r_i}[\mathsf{size}(A_{r_i})] = \E_{D \sim A}[\mathsf{size}(D)]^n$. Since $A$ is an arbitrary algorithm, taking $\log$ on both sides gives us the first inequality.
    For the second inequality,
    \begin{align*}
        \log(\mathsf{RSize}_{\gamma-\epsilon}^{\mu}(f)) &\leq \left(\frac{2}{\epsilon}\right)^2 \IC_{\gamma}^{\mu}(f) &\text{(From \cref{thm:compression})} \\
        &= \left(\frac{2}{\epsilon}\right)^2 \frac{\IC_{\gamma^{\otimes n}}^{\mu^n}(f^n)}{n} &\text{(From \cref{lem:additivity})} \\
        &\leq \left(\frac{2}{\epsilon}\right)^2 \frac{\log(\overline{\mathsf{RSize}}_{\gamma^{\otimes n}}^{\mu^n}(f^n))}{n} &\text{(From \cref{lem:treesizebound})}&\qedhere
    \end{align*}
\end{proof}

\section{Direct Product Theorems}
In this section, we prove a direct product theorem for success-conditioned information complexity. As a corollary, we will recover a direct product theorem for randomized tree size complexity. Our proof follows that of \cite{BB25}, who showed a direct product theorem for randomized query complexity.

In our work, we abstract out the properties of a \emph{cost measure} for the nodes of a tree, which generalizes the depth of a decision tree (i.e., randomized query complexity), such that the techniques of \cite{BB25} give a direct product theorem for the success-conditioned version of any cost measure which satisfies these properties. We then show that the surprisal cost measure $\mathsf{surprisal}^{\mu}(\ell)$ described earlier satisfies these properties, and thus obtain a direct product theorem for success-conditioned information complexity.

In order to highlight the properties of cost and score measures used in the proof, for most of this section we will state our proofs using generalized cost measures (which we define in \cref{sec:cost}) and score measures (which was defined in \cite{BB25}, we review it in \cref{sec:score}). At the end of this section, we will specialize our theorem statements to success-conditioned information complexity, by using surprisal cost and success probability score.

The proof of \cite{BB25} has two main steps: they first show that a certain measure called discounted score satisfies a perfect tensorization lemma, and then they show an equivalence lemma for the success-conditioned randomized query complexity and discounted score. We show that their tensorization lemma for discounted score works for the setting of our generalized cost measures, and so does their equivalence lemma. We start by describing the generalized cost and score measures, their properties that we will use, and then the discounted score.

\subsection{Score Measures}\label{sec:score}
In this section, we review the definition of score measures from \cite{BB25}. In our work, we only use the success probability score, the generalized score measures are reviewed as an abstraction to highlight the specific properties of success probability used in the proof. A function $\phi:[0,1] \rightarrow \mathbb{R}$ is a score function if it satisfies the following properties:
\begin{itemize}
    \item Boundedness: The range of $\phi$ is $\left[\frac{1}{2}, 1\right]$.
    \item Normalization: $\phi(0) = \phi(1) = 1$, $\phi\left(\frac{1}{2}\right) = \frac{1}{2}$.
    \item Symmetry: $\phi(p) = \phi(1-p)~\forall~p \in [0,1]$
    \item Monotonicity: $\phi$ is monotonically increasing on $\left[\frac{1}{2}, 1\right]$.
    \item Continuous: $\phi$ is continuous on $[0,1]$.
\end{itemize}
Then the \emph{score} of a leaf $\ell$ which occurs with non-zero probability with respect to input distribution $\mu$ on $\{0,1\}^k$ and function $f: \{0,1\}^k \rightarrow \{0,1\}$, is defined as follows:
\begin{align*}
    \mathsf{score}_f^{\mu}(\ell) &= \phi\left(\Pr_{x \sim \mu}[f(x) = 1|\ell \subseteq x]\right)
\end{align*}
where $\ell \subseteq x$ is used to denote that $x$ is consistent with $\ell$. Note that we can always modify a tree to remove leaves $\ell$ which are not consistent with any input in the support of $\mu$, without affecting the behavior of the tree. Given a randomized algorithm $R$, computing $f:\{0,1\}^k \rightarrow \{0,1\}$ on the input distribution $\mu$, we define the score of the algorithm $R$:
    \begin{align*}
        \mathsf{score}_f^{\mu}(R) &= \E_{\ell \sim L(\mu, R)}[\mathsf{score}_f^{\mu}(\ell)] &\text{(Average score over distribution $\mu$)}
    \end{align*}

\paragraph{Computing $n$ copies of $f$.} We first define the score for a leaf $\ell$ on $(\{0,1\}^{k})^{n}$ with respect to input distribution $\mu^n$ for function $f^n$. We denote $\ell = \ell_1 \times \ldots \times \ell_n$, where $\ell_i$ is the query responses on the $i^{th}$ input from $\{0,1\}^k$ for $i \in [n]$.
\begin{align*}
    \mathsf{score}_{f^n}^{\mu^n}(\ell) := \prod_{i=1}^n \mathsf{score}_f^{\mu}(\ell_i)
\end{align*}
We define the overall score of $R$ with respect to the input distribution $\mu^n$ as
\begin{align*}
    \mathsf{score}_{f^n}^{\mu^n}(R) &= \E_{\ell \sim L(\mu^n, R)}[\mathsf{score}_{f^n}^{\mu^n}(\ell)] &\text{(Overall score on distribution $\mu^n$)}
\end{align*}
Note that with this generalized formulation of score measures, the success probability score is defined using the score function $\phi_{success}(p) = \max \{p, 1-p\}$.

\subsection{Cost Measures}\label{sec:cost}
In this section, we define a generalized cost measure for each leaf and internal node of a decision tree. The cost of a leaf $\ell$ (or internal node) with respect to the input distribution $\mu$ on $\{0,1\}^k$, is denoted by $\mathsf{cost}^{\mu}(\ell)$. We say that a cost measure $\mathsf{cost}^{\mu}$ is valid if it satisfies the following properties:
\begin{itemize}
    \item Normalization: Empty leaf (no queries) has cost $0$. If $\ell = \ast^n$, then we have that $\mathsf{cost}^{\mu}(\ell) = 0$.
    \item Monotonicity: Consider an internal node $v$ of $D$ on the path from root to leaf $\ell$. Then
    \begin{align*}
        \mathsf{cost}^{\mu}(v) \leq \mathsf{cost}^{\mu}(\ell)
    \end{align*}
    So traversing down a path on the tree increases cost.
    \item Additivity: If $\ell = \ell_1 \times \ldots \times \ell_n$, where $\ell$ is a leaf on $(\{0,1\}^{k})^{n}$ and $\ell_i$ consists of the queries on the $i^{th}$ input from $\{0,1\}^k$ for $i \in [n]$, then
    \begin{align*}
        \mathsf{cost}^{\mu^n}(\ell) = \sum_{i=1}^n \mathsf{cost}^{\mu}(\ell_i),
    \end{align*}
\end{itemize}
Note that the surprisal cost measure $\mathsf{surprisal}^{\mu}(\ell)$ is a valid cost measure since we have shown that it satisfies these properties. Given a randomized algorithm $R$, computing some function $f:\{0,1\}^k \rightarrow \{0,1\}$ on an input distribution $\mu$, we will now define its expected cost:
    \begin{align*}
        \overline{\mathsf{cost}}^{\mu}(R) &= \E_{\ell \sim L(\mu, R)}[\mathsf{cost}^{\mu}(\ell)] &\text{(Average cost over distribution $\mu$)}
    \end{align*}
Moreover, we can define its score-weighted cost (analogue of success-conditioned information complexity) as follows:
\begin{align*}
    \overline{\mathsf{scost}}^{\mu}_f(R) = \frac{\E_{\ell \sim L(\mu, R)}[\mathsf{score}^{\mu}_f(\ell)\cdot \mathsf{cost}^{\mu}(\ell)]}{\mathsf{score}^{\mu}_f(R)}
\end{align*}

\subsection{Discounted Score}\label{sec:discountedscore}
Given some specific $\mathsf{score}_f^{\mu}$ and $\mathsf{cost}^{\mu}$, and a discount factor $\alpha$, the discounted score for a randomized decision tree $R$ computing $f$ on input distribution $\mu$ is defined similarly to \cite{BB25}:
\begin{align*}
    \mathsf{ds}_{f, \alpha}^{\mu}(R) = \E_{\ell \sim L(\mu, R)}[\mathsf{score}^{\mu}_f(\ell) \cdot e^{-\alpha \cdot \mathsf{cost}^{\mu}(\ell)}]
\end{align*}
The maximum discounted score of $f$ with respect to $\mu$ is
\begin{align*}
    \mathsf{DS}_{\alpha}^{\mu}(f) = \max_{R} \mathsf{ds}_{f, \alpha}^{\mu}(R)
\end{align*}
Intuitively, for any algorithm $R$, the discounted score computes the expected score (over the choice of leaves) of the algorithm, while charging it for the cost paid by the algorithm on each leaf to achieve that score. If the cost of a leaf is greater than $1/\alpha$, the algorithm is charged heavily for it, so this measure is supposed to rule out algorithms which achieve high score by paying cost significantly higher than $1/\alpha$. We show this formally in \cref{prop:expectscore} that if $R$ achieves best discounted score for discount factor $\alpha$, then the expected score of $R$ on high cost leaves decreases very fast. In general, if score-weighted complexity is low, then discounted score is high. What about the converse? Suppose for some function $f$ and score parameter $\gamma$, every algorithm which achieves score parameter at least $\gamma$ has high score-weighted complexity. Then intuitively all these algorithms have low discounted score. However, it is possible that some other algorithm (with lower score) has a better discounted score. It was shown in \cite{BB25}, that there exists a suitable discount factor $\alpha^{\ast}$, such that the maximum discounted score is achieved by an algorithm with score exactly $\gamma$. Though \cite{BB25} only considered the discounted score with the cost function being number of queries, the proof of \cref{prop:scoreds} does not use any properties of the cost function, so it holds in our setting also.
\begin{proposition}[Proposition 28 of \cite{BB25}]\label{prop:scoreds}
    For any function $f:\{0,1\}^k \rightarrow \{0,1\}$, input distribution $\mu$ and score parameter $\gamma$ (which is at least the score of the trivial guessing algorithm), there exists a discount factor $\alpha^{\ast}$ and an algorithm $R^{\ast}$ such that
    \begin{align*}
        \mathsf{ds}_{f, \alpha^{\ast}}^{\mu}(R^{\ast}) &= \mathsf{DS}_{\alpha^{\ast}}^{\mu}(f) ; &\mathsf{score}^{\mu}_f(R^{\ast}) &= \gamma
    \end{align*}
\end{proposition}

\subsection{Tensorization for Discounted Score}
In this section, we show that the discounted score, $\mathsf{DS}_{\alpha}^{\mu}(f)$, tensorizes perfectly when defined using generalized cost measures. This follows from the same proof as \cite{BB25}, on observing that in their proof the cost measure is the number of queries on an input, but their proof only uses the properties of cost measures defined in our abstraction of generalized cost measure.

\begin{lemma}[Tensorization Lemma, Lemma 5 of \cite{BB25}]\label{lem:tensorization}
    For any function $f: \{0,1\}^k \rightarrow \{0,1\}$, input distribution $\mu$ on $\{0,1\}^k$, discount factor $\alpha$ and $n \geq 1$,
    \begin{align*}
        \mathsf{DS}_{\alpha}^{\mu^n}(f^n) = \mathsf{DS}_{\alpha}^{\mu}(f)^n
    \end{align*}
\end{lemma}
\begin{proof}
    We first prove that $\mathsf{DS}_{\alpha}^{\mu^n}(f^n) \geq \mathsf{DS}_{\alpha}^{\mu}(f)^n$. Let $A$ be any randomized algorithm computing $f$. Then we define the algorithm $B$ for computing $f^n$ on $\mu^n$ by running $n$ independent copies of $A$ on the $n$ inputs. Then
    \begin{align*}
        \mathsf{ds}_{f^n, \alpha}^{\mu^n}(B) &= \E_{\ell_1, \ldots, \ell_n \sim L(\mu, A)}\left[\left(\prod_{i=1}^n \mathsf{score}^{\mu}_f(\ell_i) \right) \cdot e^{-\alpha \cdot \mathsf{cost}^{\mu^n}(\ell_1 \ldots \ell_n)}\right] \\
        &= \E_{\ell_1, \ldots, \ell_n \sim L(\mu, A)}\left[\prod_{i=1}^n \mathsf{score}^{\mu}_f(\ell_i) \cdot e^{-\alpha \cdot \mathsf{cost}^{\mu}(\ell_i)}\right] \\
        &= \prod_{i=1}^n\E_{\ell_i \sim L(\mu, A)}[\mathsf{score}^{\mu}_f(\ell_i) \cdot e^{-\alpha \cdot \mathsf{cost}^{\mu}(\ell_i)}] \\
        &= \mathsf{ds}_{f, \alpha}^{\mu}(A)^n
    \end{align*}
    Since $A$ is an arbitrary algorithm, we get that $\mathsf{DS}_{\alpha}^{\mu^n}(f^n) \geq \mathsf{DS}_{\alpha}^{\mu}(f)^n$.\\
    \\
    \noindent For the other direction, suppose $B$ is a deterministic algorithm computing $f^n$ on input distribution $\mu^n$ that has discounted score equal to $\mathsf{DS}_{\alpha}^{\mu^n}(f^n)$. Then we define an algorithm $A$ for computing $f$ on input distribution $\mu$ similar to \cite{BB25}. In particular, $A$ receives a true input $x \sim \mu$, and it simulates $B$ using a distribution $\nu$ by running it as follows: put the true input $x$ at a uniformly random location $i$ in the input to $B$, and when $B$ makes a query to the $i^{th}$ input, then $A$ queries $x$, otherwise it samples the response to the next query from the distribution $\nu$ conditioned on the query responses seen so far. Here the distribution $\nu$ is defined as follows:
    \begin{align*}
        \nu(y) = \frac{\mu^n(y)\cdot \mathsf{score}_{f^n}^{\mu^n}(\ell_y)\cdot e^{-\alpha \cdot \mathsf{cost}^{\mu^n}(\ell_y)}}{\mathsf{DS}_{\alpha}^{\mu^n}(f^n)}
    \end{align*}
    where $\ell_y$ is the leaf of $B$ corresponding to the input $y$. Let $\pi_i$ be the distribution induced on the leaves of $B$ when $A$ puts the true input in the $i^{th}$ location and then simulates $B$ as above. We can also extend $\pi_i$ to the inner nodes $v$ in $B$, where $\pi_i(v)$ is the probability that $A$ reaches the node $v$ conditioned on putting the true input in the $i^{th}$ location. Then for any leaf $\ell$, we have that
    \begin{align*}
        \pi_i(\ell) = \prod_{(v,w) \in p(\ell)} \pi_i(w|v)
    \end{align*}
    where $(v,w)$ denote the edges in the path $p(\ell)$ from the root to the leaf $\ell$. Here $\pi_i(w|v)$ denotes the probability of going to node $w$ conditioned on seeing the partial assignment of node $v$. Therefore,
    \begin{align*}
        \pi_i(w|v) =
        \begin{cases}
        \mu^n(w|v) & \text{ if node $v$ queries the $i^{th}$ input} \\
        \nu(w|v) & \text{ otherwise}
        \end{cases}
    \end{align*}
    For every leaf $\ell$ in $B$, we denote the partial assignment for the $i^{th}$ input by $\ell_i$. The discounted score of $A$ is then
    \begin{align*}
        \mathsf{ds}_{f, \alpha}^{\mu}(A) &= \E_{\ell \sim L(\mu, A)}[\mathsf{score}^{\mu}_f(\ell) \cdot e^{-\alpha \cdot \mathsf{cost}^{\mu}(\ell)}] \\
        &= \E_{i \sim [n]} \E_{\ell \sim \pi_i}[\mathsf{score}^{\mu}_f(\ell_i) \cdot e^{-\alpha \cdot \mathsf{cost}^{\mu}(\ell_i)}] \\
        &= \sum_{\ell \in B} \E_{i \sim [n]}[\pi_i(\ell)\cdot\mathsf{score}^{\mu}_f(\ell_i) \cdot e^{-\alpha \cdot \mathsf{cost}^{\mu}(\ell_i)}] \\
        &\geq \sum_{\ell \in B} \left(\prod_{i=1}^n \pi_i(\ell)\cdot\mathsf{score}^{\mu}_f(\ell_i) \cdot e^{-\alpha \cdot \mathsf{cost}^{\mu}(\ell_i)}\right)^{1/n} &\text{(Using AM-GM inequality)} \\
        &= \sum_{\ell \in B} \left(\prod_{i=1}^n \pi_i(\ell)\right)^{1/n}\left(\mathsf{score}^{\mu^n}_{f^n}(\ell) \cdot e^{-\alpha \cdot \mathsf{cost}^{\mu^n}(\ell)}\right)^{1/n} &\text{(Using additivity of cost)}
    \end{align*}
    Now we analyze the term $\prod_{i=1}^n \pi_i(\ell)$.
    \begin{align*}
        \prod_{i=1}^n \pi_i(\ell) &= \prod_{i=1}^n \prod_{(v, w) \in p(\ell)} \pi_i(w|v) \\
        &= \prod_{(v,w) \in p(\ell)} \prod_{i=1}^n \pi_i(w|v) \\
        &= \prod_{(v,w) \in p(\ell)} \mu^n(w|v) \nu(w|v)^{n-1} \\
        &= \left(\prod_{(v,w) \in p(\ell)} \mu^n(w|v)\right) \left(\prod_{(v,w) \in p(\ell)} \nu(w|v)\right)^{n-1} \\
        &= \mu^n(\ell) \nu(\ell)^{n-1} \\
        &= \frac{(\mu^n(\ell))^n\cdot (\mathsf{score}_{f^n}^{\mu^n}(\ell)\cdot e^{-\alpha \cdot \mathsf{cost}^{\mu^n}(\ell)})^{n-1}}{(\mathsf{DS}_{\alpha}^{\mu^n}(f^n))^{n-1}}
    \end{align*}
    Putting this back in the analysis for discounted score of $A$, we get
    \begin{align*}
        \mathsf{ds}_{f, \alpha}^{\mu}(A) &\geq \frac{\sum_{\ell \in B} \mu^n(\ell)\cdot \mathsf{score}_{f^n}^{\mu^n}(\ell)\cdot e^{-\alpha \cdot \mathsf{cost}^{\mu^n}(\ell)}}{(\mathsf{DS}_{\alpha}^{\mu^n}(f^n))^{1-\frac{1}{n}}} \\
        &= \left(\mathsf{DS}_{\alpha}^{\mu^n}(f^n)\right)^{\frac{1}{n}}
    \end{align*}
    Therefore, we get that $\mathsf{DS}_{\alpha}^{\mu^n}(f^n) \leq \mathsf{DS}_{\alpha}^{\mu}(f)^n$.
\end{proof}

\subsection{Equivalence Lemma for Success-Conditioned Information Complexity and Discounted Score}\label{sec:equivalence}
In this section, we show an equivalence lemma (\cref{lem:equivalence}) between success-conditioned information complexity and discounted score. Once again, the proof is written in terms of generalized score and cost measures to highlight that we use only general properties of these measures. The proof is analogous to that in \cite{BB25}, whose cost measure was the number of queries. In particular, they used the number of queries to truncate the tree on paths which have high number of queries. We do the same truncation process using the generalized cost measure of internal nodes, and then their proof works to show the equivalence lemma. The first part of the equivalence lemma informally says that if the score-weighted complexity of a function $f$ with score parameter $\gamma$ is small, then the discounted score of $f$ is small. This is shown using Jensen's inequality and the fact that $e^{-\alpha x}$ is a convex function.

The second part of this lemma says that if the score-weighted complexity of a function $f$ with score parameter $\gamma$ is large, then there is a suitable discount factor $\alpha^{\ast}$ such that the discounted score of $f$ is small. To show this, we will use \cref{prop:scoreds} which tells us that there is a discount factor $\alpha^{\ast}$ and an algorithm $R^{\ast}$ such that it achieves the maximum discounted score (with discount factor $\alpha^{\ast}$) for $f$, and it has score $\gamma$. As such, since discounted score discourages achieving high score by paying high cost, intuitively we should expect that $R^{\ast}$ should not put much of its score on its very high cost leaves. Similar to \cite{BB25}, this intuition is formalized in \cref{prop:expectscore}, to show that for any discount factor $\alpha$, if a decision tree $R$ achieves the highest discounted score, then the expected score of $R$ on leaves with cost $> kd$ is exponentially smaller (in $k$) than the expected score on leaves with cost $> d$, for $d = O(1/\alpha)$. Therefore, for our case the score-weighted cost of $R^{\ast}$ cannot be largely concentrated on the high-cost leaves. Then to show the desired upper bound on the discounted score of $R^{\ast}$, we consider a truncated tree which removes these high cost leaves, and use \cref{fact:convexityofe} to get the bound. By the monotonicity of the cost measure, this truncation is well-defined.

We now show that for any discount factor $\alpha$, the tree $R$ achieving the best discounted score does not put much of its score on its very high cost leaves, because otherwise the discounted score of $R$ can be improved by truncating the high cost leaves. Note that \cite{BB25} showed this is true for the specific choice of discount factor $\alpha^{\ast}$ and decision tree $R^{\ast}$ from \cref{prop:scoreds}. We slightly generalize their proof to establish our claim. Let $X_{>\tau}(\ell)$ be the indicator variable for the event that $\mathsf{cost}^{\mu}(\ell) > \tau$.

\begin{proposition}\label{prop:expectscore}
    For any function $f:\{0,1\}^k \rightarrow \{0,1\}$, input distribution $\mu$ and discount factor $\alpha^{\ast} \geq 0$, let $R^{\ast}$ be the randomized decision tree such that $\mathsf{ds}_{f, \alpha^{\ast}}^{\mu}(R^{\ast}) = \mathsf{DS}_{\alpha^{\ast}}^{\mu}(f)$. Fix $d = \frac{2\ln(2)}{\alpha^{\ast}}$. Then for any $m \geq 2$, we have that
    \begin{align*}
        \E_{\ell \sim L(\mu, R^{\ast})}[\mathsf{score}_f^{\mu}(\ell) \cdot X_{>md}(\ell)] \leq \left(\frac{2}{3}\right)^{m-1} \E_{\ell \sim L(\mu, R^{\ast})}[\mathsf{score}_f^{\mu}(\ell) \cdot X_{>d}(\ell)]
    \end{align*}
\end{proposition}
\begin{proof}
    From $R^{\ast}$, we first construct the truncated tree $\lfloor R^{\ast} \rfloor_\tau$, which removes the nodes $q$ of $R^{\ast}$ if $\mathsf{cost}^{\mu}(q) > \tau$. Using the fact that the score of each leaf of $\lfloor R^{\ast} \rfloor_\tau$ is at least $1/2$ and score of each leaf of $R^{\ast}$ is at most $1$,
    \begin{align*}
        \mathsf{ds}_{f, \alpha^{\ast}}^{\mu}(R^{\ast}) &= \E_{\ell \sim L(\mu, R^{\ast})}[\mathsf{score}_f^{\mu}(\ell) \cdot e^{-\alpha^{\ast} \mathsf{cost}^{\mu}(\ell)} \cdot X_{\leq \tau}(\ell)] + \E_{\ell \sim L(\mu, R^{\ast})}[\mathsf{score}_f^{\mu}(\ell) \cdot e^{-\alpha^{\ast} \mathsf{cost}^{\mu}(\ell)} \cdot X_{>\tau}(\ell)] \\
        \mathsf{ds}_{f, \alpha^{\ast}}^{\mu}(\lfloor R^{\ast} \rfloor_\tau) &\geq \E_{\ell \sim L(\mu, R^{\ast})}[\mathsf{score}_f^{\mu}(\ell) \cdot e^{-\alpha^{\ast} \mathsf{cost}^{\mu}(\ell)} \cdot X_{\leq \tau}(\ell)] + \frac{1}{2} \E_{\ell \sim L(\mu, R^{\ast})}[e^{-\alpha^{\ast}\tau} \cdot X_{>\tau}(\ell)] \\
        &\geq \E_{\ell \sim L(\mu, R^{\ast})}[\mathsf{score}_f^{\mu}(\ell) \cdot e^{-\alpha^{\ast} \mathsf{cost}^{\mu}(\ell)} \cdot X_{\leq \tau}(\ell)] + \frac{e^{-\alpha^{\ast}\tau}}{2} \E_{\ell \sim L(\mu, R^{\ast})}[\mathsf{score}_f^{\mu}(\ell) \cdot X_{>\tau}(\ell)]
    \end{align*}
     We know from the choice of $R^{\ast}$ that $\mathsf{ds}_{f, \alpha^{\ast}}^{\mu}(R^{\ast}) \geq \mathsf{ds}_{f, \alpha^{\ast}}^{\mu}(\lfloor R^{\ast} \rfloor_\tau)$. Therefore,
    \begin{align*}
        &\frac{e^{-\alpha^{\ast}\tau}}{2} \E_{\ell \sim L(\mu, R^{\ast})}[\mathsf{score}_f^{\mu}(\ell) \cdot X_{>\tau}(\ell)] \leq \E_{\ell \sim L(\mu, R^{\ast})}[\mathsf{score}_f^{\mu}(\ell) \cdot e^{-\alpha^{\ast} \mathsf{cost}^{\mu}(\ell)} \cdot X_{>\tau}(\ell)] \\
        \leq~&e^{-\alpha^{\ast}\tau} \left(\E_{\ell \sim L(\mu, R^{\ast})}[\mathsf{score}_f^{\mu}(\ell) \cdot X_{>\tau}(\ell) \cdot X_{\leq(\tau+d)}(\ell)]\right) + e^{-\alpha^{\ast}(\tau+d)} \left(\E_{\ell \sim L(\mu, R^{\ast})}[\mathsf{score}_f^{\mu}(\ell) \cdot X_{>(\tau+d)}(\ell)]\right) \\
        =~&e^{-\alpha^{\ast}\tau} \left(\E_{\ell \sim L(\mu, R^{\ast})}[\mathsf{score}_f^{\mu}(\ell) \cdot \left(X_{>\tau}(\ell)-X_{>(\tau+d)}(\ell)\right)]\right) + e^{-\alpha^{\ast}(\tau+d)} \left(\E_{\ell \sim L(\mu, R^{\ast})}[\mathsf{score}_f^{\mu}(\ell) \cdot X_{>(\tau+d)}(\ell)]\right)
    \end{align*}
    Rearranging this and using $d = \frac{2\ln(2)}{\alpha^{\ast}}$ gives us that,
    \begin{align*}
        \frac{1}{2}\E_{\ell \sim L(\mu, R^{\ast})}[\mathsf{score}_f^{\mu}(\ell) \cdot X_{>\tau}(\ell)] &\geq (1-e^{-\alpha^{\ast}d})\E_{\ell \sim L(\mu, R^{\ast})}[\mathsf{score}_f^{\mu}(\ell) \cdot X_{>(\tau+d)}(\ell)] \\
        &= \frac{3}{4}\E_{\ell \sim L(\mu, R^{\ast})}[\mathsf{score}_f^{\mu}(\ell) \cdot X_{>(\tau+d)}(\ell)]
    \end{align*}
    Therefore,
    \begin{align*}
        \E_{\ell \sim L(\mu, R^{\ast})}[\mathsf{score}_f^{\mu}(\ell) \cdot X_{>(\tau+d)}(\ell)] \leq \frac{2}{3}\E_{\ell \sim L(\mu, R^{\ast})}[\mathsf{score}_f^{\mu}(\ell) \cdot X_{>\tau}(\ell)]
    \end{align*}
    So for every $k \geq 2$,
    \begin{align*}
        \E_{\ell \sim L(\mu, R^{\ast})}[\mathsf{score}_f^{\mu}(\ell) \cdot X_{>md}(\ell)] &\leq \left(\frac{2}{3}\right) \E_{\ell \sim L(\mu, R^{\ast})}[\mathsf{score}_f^{\mu}(\ell) \cdot X_{>(m-1)d}(\ell)] \\
        &\leq \left(\frac{2}{3}\right)^{m-1} \E_{\ell \sim L(\mu, R^{\ast})}[\mathsf{score}_f^{\mu}(\ell) \cdot X_{>d}(\ell)] \qedhere
    \end{align*}
\end{proof}
\noindent As a corollary, we show that the score-weighted cost is not concentrated on the high-cost leaves.
\begin{corollary}\label{cor:expectcost}
    For any function $f:\{0,1\}^k \rightarrow \{0,1\}$, input distribution $\mu$ and discount factor $\alpha^{\ast} \geq 0$, let $R^{\ast}$ be the randomized decision tree such that $\mathsf{ds}_{f, \alpha^{\ast}}^{\mu}(R^{\ast}) = \mathsf{DS}_{\alpha^{\ast}}^{\mu}(f)$. Fix $d = \frac{2\ln(2)}{\alpha^{\ast}}$. Then
    \begin{align*}
        \E_{\ell \sim L(\mu, R^{\ast})}\left[\mathsf{score}_f^{\mu}(\ell) \cdot \mathsf{cost}^{\mu}(\ell) X_{> d}(\ell)\right] \leq 4d\E_{\ell \sim L(\mu, R^{\ast})}\left[\mathsf{score}_f^{\mu}(\ell) \cdot X_{> d}(\ell)\right]
    \end{align*}
\end{corollary}
\begin{proof}
    We use \cref{prop:expectscore} to establish the claim.
    \begin{align*}
        \E_{\ell \sim L(\mu, R^{\ast})}\left[\mathsf{score}_f^{\mu}(\ell) \cdot \mathsf{cost}^{\mu}(\ell) X_{> d}(\ell)\right] &= \sum_{m=1}^{\infty} \E_{\ell \sim L(\mu, R^{\ast})}\left[\mathsf{score}_f^{\mu}(\ell) \cdot \mathsf{cost}^{\mu}(\ell) X_{> md}(\ell) X_{\leq (m+1)d}(\ell)\right] \\
        &\leq \sum_{m=1}^{\infty} (m+1)d \E_{\ell \sim L(\mu, R^{\ast})}\left[\mathsf{score}_f^{\mu}(\ell) \cdot (X_{> md}(\ell) - X_{> (m+1)d}(\ell))\right] \\
        &= d\E_{\ell \sim L(\mu, R^{\ast})}\left[\mathsf{score}_f^{\mu}(\ell) \cdot X_{>d}(\ell)\right] + d \sum_{m=1}^{\infty} \E_{\ell \sim L(\mu, R^{\ast})}\left[\mathsf{score}_f^{\mu}(\ell) \cdot X_{> md}(\ell)\right] \\
        &\leq d\E_{\ell \sim L(\mu, R^{\ast})}\left[\mathsf{score}_f^{\mu}(\ell) \cdot X_{>d}(\ell)\right] \left(1 + \sum_{m=1}^{\infty} \left(\frac{2}{3}\right)^{m-1} \right) \\
        &= 4d\E_{\ell \sim L(\mu, R^{\ast})}\left[\mathsf{score}_f^{\mu}(\ell)\cdot X_{>d}(\ell)\right]
    \end{align*}
    where the second last inequality follows from \cref{prop:expectscore}.
\end{proof}
\noindent We now show the equivalence lemma.
\begin{lemma}[Equivalence Lemma, Lemma 6 of \cite{BB25}]\label{lem:equivalence}
    For any function $f:\{0,1\}^k \rightarrow \{0,1\}$, input distribution $\mu$, score parameter $\gamma$ and discount factor $\alpha \geq 0$,
    \begin{align*}
        \mathsf{sIC}_{\gamma}^{\mu}(f) \geq \frac{1}{\alpha} \ln \frac{\gamma}{\mathsf{DS}^{\mu}_{\alpha}(f)}
    \end{align*}
    Conversely, for every $f, \mu, \gamma$, there exists a discount factor $\alpha^{\ast}$ such that
    \begin{align*}
        \mathsf{sIC}_{\gamma}^{\mu}(f) = O\left(\frac{1}{\alpha^{\ast}} \ln \frac{\gamma}{\mathsf{DS}^{\mu}_{\alpha^{\ast}}(f)}\right).
    \end{align*}
\end{lemma}
\begin{proof}
    We start by showing the first inequality. Given a randomized algorithm $R$, we have that
    \begin{align*}
        \overline{\mathsf{scost}}_f^{\mu}(R) &= \frac{\E_{\ell \sim L(\mu, R)}[\mathsf{score}^{\mu}_f(\ell) \cdot \mathsf{cost}^{\mu}(\ell)]}{\mathsf{score}_f^{\mu}(R)} \\
        &= \frac{\sum_{\ell \in R} \Pr_{L(\mu, R)}[\ell] \cdot \mathsf{score}^{\mu}_f(\ell) \cdot \mathsf{cost}^{\mu}(\ell)}{\sum_{\ell \in R}\Pr_{L(\mu, R)}[\ell] \cdot \mathsf{score}^{\mu}_f(\ell)} \\
        \mathsf{ds}^{\mu}_{f, \alpha}(R) &= \E_{\ell \sim L(\mu, R)}[\mathsf{score}^{\mu}_f(\ell) \cdot e^{-\alpha \cdot \mathsf{cost}^{\mu}(\ell)}] \\
        &= \sum_{\ell \in R} \Pr_{L(\mu, R)}[\ell] \cdot \mathsf{score}^{\mu}_f(\ell) \cdot e^{-\alpha \cdot \mathsf{cost}^{\mu}(\ell)}
    \end{align*}
    Using Jensen's inequality (\cref{fact:jensen}) on $\overline{\mathsf{scost}}_f^{\mu}(R)$ with the convex function $e^{-\alpha \cdot x}$, we get that
    \begin{align*}
        e^{-\alpha \cdot \overline{\mathsf{scost}}_f^{\mu}(R)} &\leq \frac{\sum_{\ell \in R} \Pr_{L(\mu, R)}[\ell] \cdot \mathsf{score}^{\mu}_f(\ell) \cdot e^{-\alpha \cdot \mathsf{cost}^{\mu}(\ell)}}{\sum_{\ell \in R}\Pr_{L(\mu, R)}[\ell] \cdot \mathsf{score}^{\mu}_f(\ell)} = \frac{\mathsf{ds}^{\mu}_{f, \alpha}(R)}{\mathsf{score}_f^{\mu}(R)}
    \end{align*}
    Therefore,
    \begin{align*}
        \mathsf{DS}^{\mu}_{\alpha}(f) &= \max_R \mathsf{ds}^{\mu}_{f, \alpha}(R) \\
        &\geq \max_{R: \mathsf{score}_f^{\mu}(R) \geq \gamma} \mathsf{ds}^{\mu}_{f, \alpha}(R) \\
        &\geq \max_{R: \mathsf{score}_f^{\mu}(R) \geq \gamma}\mathsf{score}_f^{\mu}(R)\cdot e^{-\alpha \cdot \overline{\mathsf{scost}}_f^{\mu}(R)} \\
        &\geq \gamma \cdot e^{-\alpha \cdot \mathsf{sIC}_{\gamma}^{\mu}(f)}
    \end{align*}
    Taking $\ln$ on both sides gives us the first inequality. For the other direction, given $f, \mu, \gamma$, we pick discount factor $\alpha^{\ast}$ and algorithm $R^{\ast}$ from \cref{prop:scoreds} such that
    \begin{align*}
        \mathsf{ds}_{f, \alpha^{\ast}}^{\mu}(R^{\ast}) &= \mathsf{DS}_{\alpha^{\ast}}^{\mu}(f) ; &\mathsf{score}^{\mu}_f(R^{\ast}) &= \gamma
    \end{align*}
    Then setting $d = \frac{2\ln(2)}{\alpha^{\ast}}$ and using \cref{fact:convexityofe}, we have that
    \begin{align*}
        \mathsf{DS}_{\alpha^{\ast}}^{\mu}(f) &= \E_{\ell \sim L(\mu, R^{\ast})}\left[\mathsf{score}_f^{\mu}(\ell) \cdot e^{-\alpha^{\ast}\mathsf{cost}^{\mu}(\ell)}\right] \\
        &\leq \E_{\ell \sim L(\mu, R^{\ast})}\left[\mathsf{score}_f^{\mu}(\ell) \cdot e^{-\alpha^{\ast}\min\{\mathsf{cost}^{\mu}(\ell), d\}}\right] \\
        &\leq \E_{\ell \sim L(\mu, R^{\ast})}\left[\mathsf{score}_f^{\mu}(\ell) \cdot \left(1-\frac{1-e^{-\alpha^{\ast}d}}{d} \min\{\mathsf{cost}^{\mu}(\ell), d\}\right)\right] \\
        &= \mathsf{score}_f^{\mu}(R^{\ast}) - \frac{3\alpha^{\ast}}{8\ln(2)} \E_{\ell \sim L(\mu, R^{\ast})}\left[\mathsf{score}_f^{\mu}(\ell) \cdot \min\{\mathsf{cost}^{\mu}(\ell), d\}\right] \\
        &= \gamma \left(1-\frac{3\alpha^{\ast}}{8\ln(2)} \cdot \frac{\E_{\ell \sim L(\mu, R^{\ast})}\left[\mathsf{score}_f^{\mu}(\ell) \cdot \min\{\mathsf{cost}^{\mu}(\ell), d\}\right]}{\mathsf{score}_f^{\mu}(R^{\ast})}\right) \\
        & \leq \gamma \cdot e^{-\frac{3\alpha^{\ast}}{8\ln(2) \gamma} \cdot \E_{\ell \sim L(\mu, R^{\ast})}\left[\mathsf{score}_f^{\mu}(\ell) \cdot \min\{\mathsf{cost}^{\mu}(\ell), d\}\right]}
    \end{align*}
    So we now analyze $\E_{\ell \sim L(\mu, R^{\ast})}\left[\mathsf{score}_f^{\mu}(\ell) \cdot \min\{\mathsf{cost}^{\mu}(\ell), d\}\right]$, and use \cref{cor:expectcost} to show that this quantity can not be too small.
    \begin{align*}
        \E_{\ell \sim L(\mu, R^{\ast})}\left[\mathsf{score}_f^{\mu}(\ell) \cdot \mathsf{cost}^{\mu}(\ell)\right] &= \E_{\ell \sim L(\mu, R^{\ast})}\left[\mathsf{score}_f^{\mu}(\ell) \cdot \mathsf{cost}^{\mu}(\ell) X_{\leq d}(\ell)\right] + \E_{\ell \sim L(\mu, R^{\ast})}\left[\mathsf{score}_f^{\mu}(\ell) \cdot \mathsf{cost}^{\mu}(\ell) X_{> d}(\ell)\right] \\
        &\leq 4\E_{\ell \sim L(\mu, R^{\ast})}\left[\mathsf{score}_f^{\mu}(\ell) \cdot \mathsf{cost}^{\mu}(\ell) X_{\leq d}(\ell)\right] + 4d\E_{\ell \sim L(\mu, R^{\ast})}\left[\mathsf{score}_f^{\mu}(\ell) \cdot X_{> d}(\ell)\right] \\
        &= 4 \E_{\ell \sim L(\mu, R^{\ast})}\left[\mathsf{score}_f^{\mu}(\ell) \cdot \min\{\mathsf{cost}^{\mu}(\ell), d\}\right]
    \end{align*}
    Therefore, we have that
    \begin{align*}
        \mathsf{DS}_{\alpha^{\ast}}^{\mu}(f) &\leq \gamma \cdot e^{-\frac{3\alpha^{\ast}}{8\ln(2)} \cdot \frac{\E_{\ell \sim L(\mu, R^{\ast})}\left[\mathsf{score}_f^{\mu}(\ell) \cdot \mathsf{cost}^{\mu}(\ell)\right]}{4\gamma}} \leq \gamma \cdot e^{-\frac{3\alpha^{\ast}}{32\ln(2)} \cdot \mathsf{sIC}_{\gamma}^{\mu}(f)}
    \end{align*}
    Taking $\ln$ on both sides gives us the second inequality.
\end{proof}

\subsection{Proof of Direct Product Theorem}
In this section, we use the tensorization for discounted score (\cref{lem:tensorization}) and the equivalence lemma for success-conditioned information complexity and discounted score (\cref{lem:equivalence}), to establish a direct product theorem lower bound for maximum success-conditioned information complexity.
\sicdpt*
\begin{proof}
    To see this,
    \begin{align*}
        \mathsf{sIC}_{\gamma^n}^{\mu^n}(f^n) &\geq \frac{1}{\alpha^{\ast}} \ln \frac{\gamma^n}{\mathsf{DS}^{\mu^n}_{\alpha^{\ast}}(f^n)} &\text{(For any $\alpha^{\ast}$, using lower bound of \cref{lem:equivalence})} \\
        &\geq \frac{1}{\alpha^{\ast}} \ln \frac{\gamma^n}{\mathsf{DS}^{\mu}_{\alpha^{\ast}}(f)^n} &\text{(Using second part of \cref{lem:tensorization})} \\
        &= n \cdot \frac{1}{\alpha^{\ast}} \ln \frac{\gamma}{\mathsf{DS}^{\mu}_{\alpha^{\ast}}(f)} \\
        &= \Omega\left(n \cdot \mathsf{sIC}_{\gamma}^{\mu}(f)\right) &\text{(Using the $\alpha^{\ast}$ specified by the upper bound of \cref{lem:equivalence})} &\qedhere
    \end{align*}
\end{proof}

\begin{theorem}[Direct Product Theorem for $\mathsf{sIC}_{\gamma}(f)$]
    For every function $f: \{0,1\}^k \rightarrow \{0,1\}$, success probability score parameter $\frac{1}{2} \leq \gamma \leq 1$ and $n \geq 1$,
    \begin{align*}
        \mathsf{sIC}_{\gamma^n}(f^n) = \Omega\left(n\cdot\mathsf{sIC}_{\gamma}(f)\right).
    \end{align*}
\end{theorem}
\begin{proof}
    Maximizing over the choice of $\mu$ in \cref{thm:dpt}, we get the claim.
\end{proof}

\noindent We now show a distributional direct product theorem for decision tree size.

\begin{theorem}[Distributional Direct Product Theorem for Tree Size]\label{thm:treedirectproduct}
    For every function $f: \{0,1\}^k \rightarrow \{0,1\}$, distribution $\mu$ on $\{0,1\}^k$ and success probability $\frac{1}{2} \leq \gamma \leq 1$, we have that
    \begin{align*}
        \Omega(n \cdot \mathsf{IC}^{\mu}_{\gamma}(f)) - n \leq \log\left(\overline{\mathsf{RSize}}_{\gamma^n}^{\mu^n}(f^n)\right)
    \end{align*}
\end{theorem}
\begin{proof}
    Note that,
    \begin{align*}
        \log\left(\overline{\mathsf{RSize}}_{\gamma^n}^{\mu^n}(f^n)\right) &\geq \mathsf{sIC}^{\mu^n}_{\gamma^n}(f) - n &\text{(Using \cref{lem:sicub})}\\
        &= \Omega(n \cdot \mathsf{sIC}^{\mu}_{\gamma}(f)) - n &\text{(Using \cref{thm:dpt})} \\
        &= \Omega(n \cdot \mathsf{IC}^{\mu}_{\gamma}(f)) - n &\text{(Using \cref{prop:srcvsrc})}&\qedhere
    \end{align*}
\end{proof}
\noindent We can use this result to obtain a direct product theorem with decision-tree size on both sides of the inequality, using the relation between information complexity and tree size established in \cref{sec:icinquery}. A direct product theorem for randomized tree size was obtained previously by \cite{Dru12} using different techniques.

\begin{corollary}\label{cor:treedirectproduct}
    For every function $f: \{0,1\}^k \rightarrow \{0,1\}$, distribution $\mu$ on $\{0,1\}^k$, success probability $\frac{1}{2} \leq \gamma \leq 1$ and $\epsilon \in \left(0, \gamma-\frac{1}{2}\right]$, we have that
    \begin{align*}
        \frac{n\epsilon^2}{4}\Omega\left(\log\left(\mathsf{RSize}_{\gamma-\epsilon}^{\mu}(f)\right)\right) - n \leq \log\left(\overline{\mathsf{RSize}}_{\gamma^{n}}^{\mu^n}(f^n)\right) \leq n\log\left(\overline{\mathsf{RSize}}_{\gamma}^{\mu}(f)\right)
    \end{align*}
\end{corollary}
\begin{proof}
    The upper bound follows by computing $f^n$ by running $n$ independent copies of a tree computing $f$. The lower bound follows from \cref{thm:treedirectproduct} and \cref{thm:compression}.
\end{proof}

We can also use this to get a distribution-free direct product theorem for randomized decision tree size complexity.
\begin{corollary}\label{cor:fulltreedirectproduct}
    For every function $f: \{0,1\}^k \rightarrow \{0,1\}$ and success probability $\frac{1}{2} < \gamma \leq 1$, we have that
    \begin{align*}
        \max\left\{n \log(2\gamma), \frac{nc_{\gamma}}{-\ln(1-\gamma)}\Omega\left(\log\left(\mathsf{RSize}_{\gamma}(f)\right)\right) - n \right\} \leq \log\left(\overline{\mathsf{RSize}}_{\gamma^{n}}(f^n)\right) \leq n\log\left(\overline{\mathsf{RSize}}_{\gamma}(f)\right),
    \end{align*}
    where
    \begin{align*}
        c_{\gamma} = \begin{cases}
            \gamma^3 &\text{when $\gamma = \frac{2}{3}$},\\
            (2\gamma - 1)(1-\gamma)^2 &\text{when $\gamma = 1-o(1)$},\\
            (2\gamma + 1)(2\gamma-1)^2 &\text{when $\gamma = \frac{1}{2}+o(1)$}
        \end{cases}
    \end{align*}
\end{corollary}
\begin{proof}
    For the upper bound, consider any randomized decision tree $A$ of expected size $T$ which computes $f$ with success probability at least $\gamma$ (over the internal randomness). Then running $n$ independent copies of $A$ gives a tree of expected size $T^n$ which has success probability at least $\gamma^n$. Taking $\log$ on both sides gives the upper bound.

    For the lower bound, we know from \cref{cor:treedirectproduct} and \cref{lem:minmaxsize} that for $\epsilon \in \left(0, \gamma-\frac{1}{2}\right]$
    \begin{align*}
        \log\left(\overline{\mathsf{RSize}}_{\gamma^{n}}(f^n)\right) &\geq \frac{n\epsilon^2}{4}\Omega\left(\log\left(\mathsf{RSize}_{\gamma-\epsilon}(f)\right)\right) - n
    \end{align*}
    On setting
    \begin{align*}
        \epsilon = \begin{cases}
        \frac{\gamma}{4} &\text{when $\gamma = \frac{2}{3}$},\\
        1-\gamma &\text{when $\gamma = 1-o(1)$},\\
        \frac{\gamma}{2}-\frac{1}{4} &\text{when $\gamma = \frac{1}{2}+o(1)$}
        \end{cases}
    \end{align*}
    and using Chernoff bound (\cref{fact:chernoff}) to do error-reduction, we get
    \begin{align*}
        \log\left(\overline{\mathsf{RSize}}_{\gamma^{n}}(f^n)\right) \geq \begin{cases}
            \frac{n\gamma^3}{-\ln(1-\gamma)}\Omega\left(\log\left(\mathsf{RSize}_{\gamma}(f)\right)\right) - n &\text{when $\gamma = \frac{2}{3}$},\\
            \frac{n(2\gamma - 1)(1-\gamma)^2}{-\ln(1-\gamma)}\Omega\left(\log\left(\mathsf{RSize}_{\gamma}(f)\right)\right) - n &\text{when $\gamma = 1-o(1)$},\\
            \frac{n(2\gamma + 1)(2\gamma-1)^2}{-\ln(1-\gamma)}\Omega\left(\log\left(\mathsf{RSize}_{\gamma}(f)\right)\right) - n &\text{when $\gamma = \frac{1}{2}+o(1)$}
        \end{cases}
    \end{align*}
    For the other lower bound, we know that there are $2^n$ possible outputs for $f^n$. Choose a distribution $\mu^{\ast}$ on inputs $x$ to $f^n$ such that each possible output of $f^n$ has equal probability under $\mu^{\ast}$. Then for any randomized decision tree $R$ computing $f^n$, we have that
    \begin{align*}
        \gamma^n &\leq \min_{x \in (\{0,1\}^m)^n} \Pr_{D \sim R}[D(x) = f^n(x)] \leq \E_{x\sim \mu^{\ast}} \E_{D \sim R}[\mathbb{I}_{D(x) = f^n(x)}] \\
        &\leq \E_{D \sim R}\left[\frac{\mathsf{size}(D)}{2^n}\right] = \frac{1}{2^n} \overline{\mathsf{RSize}}(R)
    \end{align*}
    Minimizing over choice of $R$, we get that
    \begin{align*}
        \overline{\mathsf{RSize}}_{\gamma^{n}}(f^n) &\geq (2\gamma)^n \qedhere
    \end{align*}
\end{proof}

\section{Information and Query Complexity with a Parity Gadget}

In this section, we relate the information complexity to the depth of a randomized query algorithm. In particular, for a function $f$ and an input distribution $\mu$, we consider the information complexity of the function $f \circ \oplus_2$ with respect to a suitable distribution $\nu$. This distribution is sampled as follows: sample a string $x \sim \mu$, and then for each bit $x_i$, sample two uniformly random bits conditioned on their parity being equal to $x_i$. We show that $\IC^\nu_\gamma(f \circ \oplus_2)$ is equal to a measure $\mathsf{M}^{\mu}_{\gamma}(f)$, which is defined as follows:
\begin{align*}
    \mathsf{M}^{\mu}_{\gamma}(f) := \min_{R: \mathsf{success}^{\mu}_f(R) \geq \gamma} \overline{\mathsf{depth}}^{\mu}(R) + \info^{\mu}(R)
\end{align*}
That is, the measure $\mathsf{M}^{\mu}_{\gamma}(f)$ asks an algorithm to minimize the sum of its information cost and expected depth (over both the internal randomness and the distribution $\mu$). Therefore, $\mathsf{M}^{\mu}_{\gamma}(f)$ is large if either the distributional query complexity or the information complexity of $f$ is large. We first show that in fact, $\mathsf{M}^{\mu}_{\gamma}(f) = \Theta\left(\overline{\mathsf{R}}^{\mu}_{\gamma}(f)\right)$.

\begin{proposition}\label{prop:mvsrdoublebar}
    For any function $f: \{0,1\}^k \rightarrow \{0,1\}$, distribution $\mu$ on $\{0,1\}^k$ and success probability $\gamma$,
    \begin{align*}
        \mathsf{M}^{\mu}_{\gamma}(f) = \Theta\left(\overline{\mathsf{R}}^{\mu}_{\gamma}(f)\right)
    \end{align*}
\end{proposition}
\begin{proof}
    The lower bound on $\mathsf{M}^{\mu}_{\gamma}(f)$ follows from definition. For the upper bound, consider a deterministic decision tree $D$. For $d \geq 1$, let $L_d$ be the set of leaves of $D$ at depth $d$ and $p_\ell$ be the probability of seeing leaf $\ell$ under the distribution $\mu$. Then,
    \begin{align*}
        \overline{\mathsf{depth}}^{\mu}(D) - \info^{\mu}(D) = \sum_{d \geq 1} \sum_{\ell \in L_d} p_{\ell} (d + \log(p_{\ell}))
    \end{align*}
    We now look at the function $dx + x \log(x)$ for $x \in (0, 1]$ for a fixed constant $d \in \mathbb{N}$. The value of this function decreases for $x \in \left(0, \frac{1}{2^{d+1}}\right]$, and then increases for $x \in \left( \frac{1}{2^{d+1}}, 1\right]$. So the minimum value of $dx + x \log(x)$ is $\frac{-1}{2^{d+1}}$, which occurs at $x = \frac{1}{2^{d+1}}$. Therefore,
    \begin{align*}
        \overline{\mathsf{depth}}^{\mu}(D) - \info^{\mu}(D) \geq - \sum_{d \geq 1}\sum_{\ell \in L_d} \frac{1}{2^{d+1}} = -\frac{1}{2} \sum_{d \geq 1} \frac{|L_d|}{2^d} \geq -\frac{1}{2}
    \end{align*}
    Here, the last inequality follows from Kraft's inequality (\cref{fact:kraft}). As such,
    \begin{align*}
        \overline{\mathsf{depth}}^{\mu}(D) + \info^{\mu}(D) \leq 2\cdot \overline{\mathsf{depth}}^{\mu}(D) + \frac{1}{2} \leq 3\cdot \overline{\mathsf{depth}}^{\mu}(D) &\qedhere
    \end{align*}
\end{proof}

\begin{lemma}\label{lem:ictoquery}
    For any function $f: \{0,1\}^k \rightarrow \{0,1\}$, distribution $\mu$ on $\{0,1\}^m$, success probability $\gamma$, consider the function $f \circ \oplus_2$ and $\nu$ be the distribution sampled using the process above. Then
    \begin{align*}
        \IC^\nu_\gamma(f \circ \oplus_2) = \Theta\left(\overline{\mathsf{R}}^{\mu}_{\gamma}(f)\right)
    \end{align*}
\end{lemma}
\begin{proof}
    We will show that $\IC^\nu_\gamma(f \circ \oplus_2) = \mathsf{M}^{\mu}_{\gamma}(f)$, and then the claim follows from \cref{prop:mvsrdoublebar}. We first show the upper bound $\IC^\nu_\gamma(f \circ \oplus_2) \leq \mathsf{M}^{\mu}_{\gamma}(f)$. For every input $y = y_1 \ldots y_{2k}$ to $f \circ \oplus_2$, we will denote $x_i = y_{2i-1} \oplus y_{2i}$. Given an algorithm $A$ for computing $f$ with success probability $\gamma$, we define an algorithm $B$ for computing $f \circ \oplus_2$ which simulates $A$. For each query to $x_i$ of $A$, $B$ makes a query to both $y_{2i-1}$ and $y_{2i}$, and sets $x_i = y_{2i-1} \oplus y_{2i}$. By the choice of distribution $\nu$, we see that $B$ simulates $A$ on the distribution $\mu$, therefore the success probability of $B$ is at least $\gamma$. Now consider a deterministic tree $D$ in the support of $A$. We identify each leaf $\ell$ of $D$ with a partial assignment $p_\ell$ on $x$. For each such partial assignment $p_\ell$, there are $2^{|p_\ell|}$ partial assignments on $y$. Moreover, each of these partial assignments are equally likely under $\nu$, and their probability is equal to $\frac{\Pr_{\mu}[p_\ell]}{2^{|p_\ell|}}$. We now compute the information cost of the corresponding deterministic tree $D'$ in the support of $B$.
    \begin{align*}
        \info^{\nu}(D') &= \sum_{\ell \in \mathsf{supp}(D')} \Pr_{\nu}[\ell] \log\left(\frac{1}{\Pr_{\nu}[\ell]}\right) \\
        &= \sum_{\ell \in \mathsf{supp}(D)} 2^{|p_{\ell}|}\frac{\Pr_{\mu}[\ell]}{2^{|p_{\ell}|}} \log\left(\frac{2^{|p_{\ell}|}}{\Pr_{\mu}[\ell]}\right) \\
        &= \sum_{\ell \in \mathsf{supp}(D)} \Pr_{\mu}[\ell] |p_{\ell}| +  \sum_{\ell \in \mathsf{supp}(D)} \Pr_{\mu}[\ell]\log\left(\frac{1}{\Pr_{\mu}[\ell]}\right) \\
        &= \overline{\mathsf{depth}}^{\mu}(D) + \info^{\mu}(D)
    \end{align*}
    Therefore, we have that
    \begin{align*}
        \info^{\nu}(B) &= \overline{\mathsf{depth}}^{\mu}(A) + \info^{\mu}(A)
    \end{align*}
    We get the upper bound by minimizing over the choice of algorithm $A$. We now show the lower bound $\IC^\nu_\gamma(f \circ \oplus_2) \geq \mathsf{M}^{\mu}_{\gamma}(f)$. Suppose $B$ is an algorithm computing $f \circ \oplus_2$ with success probability $\gamma$ on the distribution $\nu$. For every bit $x_i$ of $x$, we denote the corresponding bits of $y$ by $y_{i,1}, y_{i,2}$. Then we define an algorithm $A$ to compute $f$ on the distribution $\mu$ by simulating $B$. For every query $y_{i,b}$ of $B$, if this is the first query to the corresponding bit $x_i$ of $x$, then $A$ samples a uniformly random bit and responds to $B$. Otherwise, $A$ queries the corresponding bit $x_i$ of $x$, and simulates the query response of $B$ by outputting the parity of $x_i$ and the bit $y_{i, \overline{b}}$ it sampled earlier. Then we can see that $A$ simulates $B$ on the distribution $\nu$, so it has success probability $\gamma$. We now consider a deterministic tree $D$ in the support of $B$. For any leaf $\ell$ in the support of $D$, use $\ell_x$ to denote the corresponding leaf in the randomized tree $R$ in support of $A$. The probability of reaching the leaf $\ell_x$ in $R$ (over the randomness of $R$ as well as the input) is at most $\frac{\Pr_\mu[\ell_x]}{2^{|\ell_x|}}$, where $\Pr_\mu[\ell_x]$ is the probability of sampling an input $x$ from $\mu$ that is consistent with $\ell_x$. Then
    \begin{align*}
        \info^{\nu}(D) &= \sum_{\ell \in \mathsf{supp}(D)} \Pr_{\nu}[\ell] \log\left(\frac{1}{\Pr_{\nu}[\ell]}\right) \\
        &\geq \sum_{\ell_x \in \mathsf{supp}(R)} \Pr[\ell_x]\log\left(\frac{2^{|\ell_x|}}{\Pr_{\mu}[\ell_x]}\right) \\
        &= \sum_{\ell_x \in \mathsf{supp}(R)}\Pr[\ell_x]\left[|\ell_x| + \log\left(\frac{1}{\Pr_{\mu}[\ell_x]}\right)\right] \\
        &= \overline{\mathsf{depth}}^{\mu}(R) + \E_{T \sim R} \E_{\ell_x \sim T}\left[\log\left(\frac{1}{\Pr_{\mu}[\ell_x]}\right)\right] \\
        &= \overline{\mathsf{depth}}^{\mu}(R) + \info^{\mu}(R)
    \end{align*}
    Here, the second inequality follows because $-\log(x)$ is a decreasing function. Therefore,
    \begin{align*}
        \info^{\nu}(B) &\geq \overline{\mathsf{depth}}^{\mu}(A) + \info^{\mu}(A)
    \end{align*}
    Then, we get the lower bound by minimizing over the choice of algorithm $B$.
\end{proof}

\ictoquery*
\begin{proof}
    For the upper bound, consider a randomized algorithm $A$ which computes $f \circ \oplus_2$ with worst-case success probability $2/3$. The information cost of this algorithm for any distribution $\mu$ is at most the worst-case query complexity (over both inputs and internal randomness) of the algorithm. Moreover, the latter quantity is clearly bounded above by twice the worst-case randomized query complexity of $f$. Therefore, minimizing over the choice of algorithms $A$ and maximizing over the choice of distributions $\mu$ gives us the upper bound.

    For the lower bound, we know from \cref{lem:ictoquery} and \cref{lem:minmaxquery} that
    \begin{align*}
        \max_{\mu_1}\IC^{\mu_1}(f \circ \oplus_2) \geq \max_{\mu_2} \overline{\mathsf{R}}^{\mu_2}(f) \geq \Omega(\mathsf{R}(f)) &\qedhere
    \end{align*}
\end{proof}

\if\anon1
\else
\phantomsection\addcontentsline{toc}{section}{Acknowledgements}
\section*{Acknowledgements}
No use of LLMs was made in this work, except for spotting minor typing errors after the paper was written. The first draft of this work was finished in early April 2026. We thank the anonymous reviewers of FOCS 2026 for their helpful comments.

AA was supported in part by a Cheriton Graduate Scholarship from the School of Computer Science. EB is supported by an NSERC Discovery grant. SB is supported in part by the Natural Sciences and
Engineering Research Council of Canada (NSERC), DGECR-2019-00027 and
RGPIN-2019-04804.\footnote{Cette recherche a été financée par le
Conseil de recherches en sciences naturelles et en génie du Canada
(CRSNG), DGECR-2019-00027 et RGPIN-2019-04804.}

\fi

\printbibliography

\if\conf0
\appendix

\section{Another per-leaf measure and why it does not work.}
For establishing the direct product theorem for success-conditioned information complexity, we used the fact that the information cost of a tree is equal to its expected surprisal cost of the leaves, where the surprisal cost satisfies certain natural properties we can hope for from a reasonable cost measure. However, using the definition of information cost, we could equally well have defined another potential per-leaf measure, whose expectation is equal to the information. Formally, for a leaf $\ell$ of a decision tree $D$, given a distribution $\mu$, we define the measure
\begin{align*}
    \mathsf{meas}^{\mu}(\ell) = H(X) - H(X|\ell)
\end{align*}
Then clearly we have that
\begin{align*}
    \E_{\ell \sim \mu}[\mathsf{meas}^{\mu}(\ell)] = H(X) - H(X|L) =  \info^{\mu}(D)
\end{align*}
However, it turns out we can not use a success-conditioned version of this measure to establish our direct product theorem. This is because $\mathsf{meas}^{\mu}(\ell)$ is not necessarily non-negative. Consider for example the distribution $\mu$ on $\{0,1\}^{n+1}$ which puts probability mass $1/2$ on the input $0^{n+1}$, and probability mass $1/2n$ on the input $1x$ where $x$ contains a single $1$ in it. So we have that $H(X) = \frac{1}{2} \log(n) + 1$. Now, suppose the tree $D$ queries the first bit of the input sampled from this distribution. Conditioned on seeing the first bit equal to $1$, the entropy $H(X|\ell)$ is now $\log(n)$. Therefore, $\mathsf{meas}^{\mu}(\ell) = H(X) - H(X|\ell) = 1 - \frac{1}{2} \log(n) < 0$.

Essentially, the negativity of the cost of a leaf is a problem because then it is no longer clear how to establish an equivalence lemma between the discounted score and the success-conditioned complexity. This is because we use the non-negativity property of the cost measure to use \cref{fact:convexityofe}, and get an upper bound on the discounted score in terms of the success-conditioned complexity. Moreover, if we try to truncate the leaves with a highly negative cost, the discounted score could potentially become much smaller.
\fi

\end{document}